\documentclass[journal,10pt]{IEEEtran}
\usepackage{graphicx,epstopdf,multirow,multicol,booktabs,pgfplots}

\usepackage{amsmath,bm}
\usepackage{amssymb,amsthm}
\usepackage[T1]{fontenc}
\usepackage[colorlinks,citecolor=blue]{hyperref}
\allowdisplaybreaks
\usepackage{setspace}

\usepackage{amssymb}
\usepackage{amsmath}
\usepackage{url}
\usepackage{eucal}
\usepackage{accents}
\newcommand{\ubar}[1]{\underaccent{\bar}{#1}}
\usepackage{hyperref}
\allowdisplaybreaks

\theoremstyle{remark}
\newtheorem{proposition}{Proposition}

\newtheorem{corollary}{Corollary}
\newtheorem{lemma}{Lemma}

\newtheorem{definition}{Definition}

\pgfplotsset{compat=1.14}

\begin{document}
\title{State and Parameter Estimation for Natural Gas Pipeline Networks using Transient State Data}

\author{Kaarthik Sundar$^{*}$ \and Anatoly Zlotnik$^{\dagger}$
\thanks{$^{*}$Center for Nonlinear Studies, Los Alamos National
Laboratory, Los Alamos, New Mexico, USA. E-mail: \texttt{kaarthik01sundar@gmail.com}}\;
\thanks{$^{\dagger}$Applied Mathematics and Plasma Physics Division, Los Alamos National
Laboratory, Los Alamos, New Mexico, USA. E-mail: \texttt{azlotnik@lanl.gov}}\;
}
%\thanks{Manuscript received April 19, 2005; revised August 26, 2015.}}

\markboth{Journal of \LaTeX\ Class Files,~Vol.~14, No.~8, November~2016}%
{Sundar \MakeLowercase{\textit{et al.}}: State and Parameter Estimation for Natural Gas Pipeline Networks using Transient State Data}

\maketitle

\begin{abstract}
We formulate two estimation problems for pipeline systems in which measurements of compressible gas flow through a network of pipes is affected by time-varying injections, withdrawals, and compression.  We consider a state estimation problem that is then extended to a joint state and parameter estimation problem that can be used for data assimilation.  In both formulations, the flow dynamics are described on each pipe by space- and time-dependent density and mass flux that evolve according to a system of coupled partial differential equations, in which momentum dissipation is modelled using the Darcy-Wiesbach friction approximation. These dynamics are first spatially discretized to obtain a system of nonlinear ordinary differential equations on which state and parameter estimation formulations are given as nonlinear least squares problems.  A rapid, scalable computational method for performing a nonlinear least squares estimation is developed. Extensive simulations and computational experiments on multiple pipeline test networks demonstrate the effectiveness of the formulations in obtaining state and parameter estimates in the presence of measurement and process noise. 
%We formulate a state estimation and joint state and parameter estimation problems for the dynamics of compressible gas flow through a network of pipelines with time-varying injections, withdrawals, and compression. The flow dynamics are described on each pipeline by space- and time-dependent density and mass flux that evolve according to a system of coupled partial differential equations; momentum dissipation is modelled using the Darcy-Wiesbach friction approximation. The dynamics is first spatially discretized to obtain a system of nonlinear ordinary differential equations on which state  and parameter estimation problems are formulated as nonlinear least squares problems. A rapid, scalable computational method for performing a nonlinear least squares estimation is developed. Extensive computational experiments on multiple gas networks corroborate the effectiveness of the formulations in obtaining state and parameter estimates. 
\end{abstract}

\IEEEpeerreviewmaketitle

\section{Introduction} \label{sec:intro}
The push towards cleaner electric power sources and the increasing supply of natural gas in the United States have led to a significant increase in the capacity and number of active gas-fired electric generators \cite{Lyons2013}. Such generators are used in the power grid to provide base load as well as to respond quickly to balance out fluctuations in electricity production by uncontrollable renewable sources such as wind and solar \cite{Rinaldi2001,Li2008}, resulting in significant and rapid variation in natural gas consumption. This in turn significantly impacts the pressure and the flow throughout the associated natural gas transmission network. These conditions contrast with historical withdrawals from gas transmission systems, which were more predictable and exhibited far slower variations.  Earlier methods for evaluating system capacities and solving optimization problems for natural gas transmission systems utilized steady-state gas flow models for which the state equations are algebraic \cite{Wong1968,Rothfarb1970}. However, the growing and increasingly intermittent dispatch of gas-fired power plants causes rapid changes in their gas consumption, which renders the steady-state assumption no longer representative of realistic operating conditions. Taking into account the time-varying natural gas consumption becomes even more important in the context of state and parameter estimation problems associated with the flow of compressible gas in large-scale pipeline networks. Many related problems in gas pipeline monitoring, leak detection, predictive simulation, and the like require precise information about the instantaneous network state \cite{Pal1991}. In practice, it would be prohibitively expensive to place pressure and flow meters everywhere throughout a pipeline system, which motivates the development of rapid and scalable state  and parameter estimation techniques that can be applied using transient measurements of pressure and gas-withdrawals obtained at the few fixed locations in the network where they are available. 

The transient flow of natural gas through a pipeline can be represented by the Euler equations for compressible gas flow in one dimension \cite{Osiadacz1984}, which is a system of coupled partial differential equations (PDEs). Simulating transient flows in pipelines on a scale of thousands of miles is itself problematic. Many methods are proposed in the literature to do so \cite{Grundel2013} and in fact, this is still an active area of research \cite{Chapman2005}. This difficulty of characterizing pipeline network dynamics presents challenges for estimation. The well understood approach of linearizing the PDEs around the steady-state mass flow rates and pressures has been used in the literature to obtain transfer function and state space models for the gas network dynamics \cite{Kralik1984, Reddy2006, Reddy2011, Alamian2012}.  These models are in turn used for state estimation using the traditional techniques available for linear systems.  But, as remarked previously, the emerging influence of gas-fired power plants causes a wide range of transient phenomena, which cause flows throughout the supplying pipeline networks to deviate substantially from the behavior approximated by steady-state models. Furthermore, the presence of bounds on the state variables complicates matters by adding additional algebraic constraints, making the adaptation of traditional Kalman filter-based techniques more difficult.  Therefore, there is a compelling need to develop estimation techniques for gas pipeline systems that use truly transient models, respect constraints on system states and actuation, and which can be applied to large systems with arbitrary network structures.

In this article, we formulate and solve the state estimation and joint state and parameter estimation problems for the dynamics of compressible gas flow through a pipeline network.   Our approach closely resembles the Moving Horizon Estimation (MHE) technique \cite{Rao2003} that is used for constrained state estimation for nonlinear systems. It differs from MHE techniques is the sense that time-periodicity conditions are imposed on the state variables.  This is done to provide additional structure to the dynamics in order to render the underlying approximation of the PDE constraints well-posed, as discussed in more detail in section \ref{subsec:network_model}. In contrast to previous estimation approaches proposed in the literature on pipeline systems, which relied on linearization techniques \cite{Kralik1984,Reddy2006,Alamian2012}, we approximate the system of PDEs using a new nonlinear control system model, the reduced network flow (RNF), derived from a model reduction of the gas network dynamics \cite{Grundel2013,Zlotnik2015,Zlotnik2015a}. The resulting RNF is a system of implicit nonlinear differential algebraic equations (DAEs) with time-varying injections and compression as control inputs; time-varying withdrawals and Darcy-Wiesbach friction factor as the parameters; and density and mass flow rates as states. The derivative terms in this RNF are further approximated using a finite difference scheme to obtain a set of nonlinear algebraic equations that approximate the underlying PDE dynamics with proven fidelity \cite{Zlotnik2015a,Dyachenko2017}.  The RNF scheme enables the estimation problems to be formulated as nonlinear programs (NLPs), and specifically as nonlinear least squares problems subject to nonlinear algebraic equations.  Such problems can be solved to local optimality using standard gradient descent techniques. The estimation formulations and algorithms are demonstrated using extensive simulation and computational experiments on a single pipe, a $4$-junction test instance, and a $25$-junction test instance, respectively.

The rest of the article is organized as follows. We first present the PDE model for the flow of compressible gas through a pipeline followed by the network flow-based model for a gas pipeline network with controllable actuators. The network flow model is a PDE-based model which is further reduced to a nonlinear DAE system and subsequently represented by an ordinary differential equation (ODE) system model via spatial discretization. Then, the mathematical assumptions on the ODE system are presented, and monotonicity properties of the actuated pipeline system model are verified. We then present implications of these assumptions on the estimation problem formulation and derive properties of the resulting state estimates.  We then formulate the state estimation and state and parameter estimation problems. This is followed by extensive computational experiments, conclusions, and discussion of promising future research directions.

\section{Modeling} \label{sec:model}
\subsection{Gas pipeline dynamics} \label{subsec:dynamics}
The flow of compressible gas within a horizontal pipe with slow transients that do no excite waves or shocks is adequately described using one-dimensional Euler equations \cite{Thorley1987} given by 
\begin{subequations}
\begin{flalign}
& \partial_t \rho + \partial_x \varphi = 0, & \label{eq:mass} \\
& \partial_t \varphi + \partial_x(\rho v^2) + \partial_x p + \rho g \sin \theta = -\frac{\lambda}{2D} \frac{\varphi |\varphi|}{\rho}. & \label{eq:momentum}
\end{flalign}
\label{eq:pde}
\end{subequations}
It is well-known in the literature \cite{Thorley1987} that the system of equations in \eqref{eq:pde} is hyperbolic. The Eq. \eqref{eq:mass} and \eqref{eq:momentum} are the mass and momentum balance equations, respectively. The variables $v, \rho, \varphi$, and $p$ in the Eq. \eqref{eq:pde} represent the instantaneous gas velocity, density, mass flux, and pressure, respectively. $v, \rho,$ and $\varphi$ are related as $\varphi = \rho v$ and all the variables are defined on the domain $[0,\ell] \times [0,T]$ where $\ell$ represents the length of the pipeline. The parameters in Eq. \eqref{eq:pde} are the Darcy-Wiesbach friction factor $\lambda$, acceleration due to gravity $g$, the pipe angle $\theta$, and pipe diameter $D$. The term on the right hand side of \eqref{eq:momentum} aggregates the friction effects. We assume that the pressure $p$ and density of the gas $\rho$, satisfy the equation of state $p = a^2 \rho$ with $a^2 = ZRT$, where $a$, $Z$, $R$, and $T$, are the speed of sound, gas compressibility factor, ideal gas constant, and constant temperature, respectively. 

The Eq. \eqref{eq:momentum} is valid in the regime when changes in the boundary conditions are sufficiently slow to not excite propagation of sound waves \cite{Chertkov2015}. Formally, the term $\partial_t \varphi$ in Eq. \eqref{eq:momentum} is much smaller than $\partial_x(\rho v^2) + \partial_x p$. The ratio of the pressure gradient term, $\partial_x p$ to the term $\partial_t \varphi$ is typically on the order of $0:0.001$ \cite{Osiadacz1984}. Hence, we omit the term $\partial_t \varphi$ from Eq. \eqref{eq:momentum}. Furthermore,  the flow velocities are much smaller than the speed of sound $a$; hence, the gas advection term $\partial_x(\rho v^2)$ is also omitted. In addition, we also assume that (i) the pipeline is level and (ii) the gas temperature, composition, and compressibility are uniform throughout the system.  These assumptions allow, respectively, the removal of the gravity term $\rho g \sin \theta$ and approximation of the equation of state by the ideal, constant, linear relationship $p = a^2 \rho$ between pressure and density. Throughout the rest of the article, we shall use density and pressure interchangeably because of the constant linear relationship between these quantities.  We make these assumptions here in order to simplify the formulation explored in this initial study of the class of transient pipeline estimation problems.  The assumptions can be relaxed, but we leave for future work the incorporation of effects caused by altitude, temperature, and composition changes, as well as use of more realistic equation of state relations that account for gas compressibility \cite{Lee66,Menon05}.  Using the aforementioned observations and assumptions, Eq. \eqref{eq:pde} can be rewritten as 
\begin{flalign}
& \partial_t \rho + \partial_x \varphi = 0 \quad \text{and} \quad a^2 \partial_x \rho = -\frac{\lambda}{2D} \frac{\varphi |\varphi|}{\rho}. & \label{eq:pde_1}
\end{flalign}

Unlike the system of equations in \eqref{eq:pde}, which is hyperbolic, the above approximated system in Eq. \eqref{eq:pde_1} is a parabolic system (an interested reader is referred to Appendix \ref{app:eig} for the eigenvalue equation of the system given in Eq. \eqref{eq:pde_1}).

Numerous studies have been done to validate and verify the use of equations \eqref{eq:pde_1}  to approximate equations \eqref{eq:pde}, with particular focus on the effect of omitting the term $\partial_t \varphi$ from Eq. \eqref{eq:momentum}.  It is well known that this term is necessary to accurately represent compressible gas flow in the regime of fast transients \cite{Thorley1987}, which could occur for example when a valve is opened or closed quickly so as to cause a pressure wave to propagate through a pipe.  In this paper, our focus is on the normal operating regime of slowly-varying flows in large-scale gas transmission pipelines.  To understand the mathematical representation of compressible gas flow in a pipe in the physical regime with slowly-varying boundary conditions, we refer the reader to the canonical study by Osiadacz \cite{Osiadacz1984}, in which the magnitude of all of the terms in \eqref{eq:pde} are compared empirically for several example numerical simulations.  That study concludes that when modeling flow in this regime, the effect on the solution of omitting the flux derivative term $\partial_t \varphi$ is negligible.  This has been confirmed in more recent studies as well.  An empirical study of the magnitude of all of the terms in \eqref{eq:pde} was done for the regimes of fast and slow transients \cite{Dyachenko2017}, which resulted in the same conclusions as those of Osiadacz.  The implication is that for models that represent flows that change on the time-scale of hours, rather than minutes, omitting the flux derivative term $\partial_t \varphi$ is acceptable. 

The negligible effect of omitting $\partial_t \varphi$ has been shown in other recent studies of optimal control of gas pipelines on a 24-hour time horizon.  An example optimal control problem was solved using a method that omitted this term \cite{Mak2016}, and in the same study the solution was compared to one for the same problem that was obtained using a method that included this term \cite{Zlotnik2015}, with each approach producing essentially the same solution.  While on the time-scale of minutes the change in boundary conditions in this example is slow, the change in the boundary conditions and the solution is quite significant on the time-scale of hours.  We encourage the reader to examine these previous studies.  The estimation problems examined in the present study are posed using data and modeling on a 24-hour time horizon, so arguably the simplified equations in \eqref{eq:pde_1} may be applied.

Throughout the rest of this article, the gas flow dynamics on a pipe are represented using Eq. \eqref{eq:pde_1}.  This equation has a unique solution when the initial conditions and boundary conditions consisting of one of $\rho(0,t) = \ubar{\rho}(t)$ or $\varphi(0,t) = \ubar{\varphi}(t)$ and one of $\rho(\ell,t) = \bar{\rho}(t)$ and $\varphi(\ell,t) = \bar{\varphi}(t)$ are specified for the pipe. To more conveniently represent the dynamics in Eq. \eqref{eq:pde_1}, and create a better numerically conditioned problem, we first apply the dimensional transformations
\begin{flalign} \label{eq:nondim}
&\hat{t}=\frac{t}{\ell_0/a}, \quad \hat{x}=\frac{x}{\ell_0}, \quad \hat{\rho}=\frac{\rho}{\rho_0}, \quad \hat{\varphi}=\frac{\varphi}{a\rho_0},&
\end{flalign}
where $\ell_0$ and $\rho_0$ are nominal length and density, to yield the non-dimensional gas dynamics on a pipeline,
\begin{flalign}
& \partial_t \rho + \partial_x \varphi = 0 \quad \text{and} \quad \partial_x \rho = -\frac{\lambda \ell_0}{2D} \frac{\varphi |\varphi|}{\rho}. & \label{eq:pde_nondim}
\end{flalign}
The hat symbol in the above non-dimensional equations have been omitted for readability. Alternately, we shall rewrite the Eq.  \eqref{eq:pde_nondim} as follows:
\begin{flalign}
& \partial_t \rho + \partial_x \varphi = 0 \quad \text{and} \quad \varphi + f(t, \rho, \partial_x \rho) = 0 & \label{eq:pde_nondim_f}
\end{flalign}
where, 
\begin{flalign}
& f(t, \rho, \partial_x \rho) = \operatorname{sgn}(\partial_x \rho) \sqrt{\left| - \rho\cdot\frac{2D}{\lambda \ell_0}\cdot\partial_x \rho\right|}. & \label{eq:f}
\end{flalign}
In Eq. \eqref{eq:f}, ``$\operatorname{sgn}(\cdot)$'' denotes the signum function.  In the context of general flows on networks, the function $f(\cdot,\cdot,\cdot)$ in Eqs. \eqref{eq:pde_nondim_f} and \eqref{eq:f} is also referred to as the ``dissipation function'' \cite{Vuffray2015, Zlotnik2016}.  

The friction caused by high pressure turbulent flow through each pipe causes the pressure in the pipeline to gradually decrease along the direction of flow. Gas compressors are used to boost line pressure to meet the minimum pressure requirement for delivery to customers.  We model compressor stations as controllable actuators that can be used to manipulate the state of the gas transmission system by, for example, modulating the compression ratio or discharge pressure at the station level. Because the size of the compressor station is small relative to the length of a pipeline, we represent compressor action as a multiplicative increase in the density at a point $x = c$ with conservation of flow \textit{i.e.}, $\rho(c^+, t) = \alpha(t) \cdot \rho(c^-,t)$ and $\varphi(c^+,t) = \varphi(c^-,t)$ where, $\alpha(t)$ denotes the time-dependent compression ratio between suction (intake) and discharge (outlet) pressure.  While this modeling approach represents compressors as point objects, it is straightforward to represent compressors as node-connecting elements, e.g., as a short pipe with a nodal compressor object at one end.

\subsection{Dynamics of gas flow for a network} \label{subsec:network_model}
A gas transmission pipeline network consists of pipes (edges) interconnected at junctions (nodes) where the gas flow can be compressed, withdrawn from, or injected into the system.  We model the gas pipeline network as a connected directed graph $\mathcal G = (\mathcal V, \mathcal E)$ where $\mathcal V$ and $\mathcal E$ represent the set of junctions and the set of pipelines connecting any pair of junctions, respectively. We use $(i,j) \in \mathcal E$ to denote the pipeline that connects the junctions $i, j\in \mathcal V$. Let $\rho_{ij}$ and $\varphi_{ij}$ denote the instantaneous density and mass flux, respectively, within the edge $(i,j) \in \mathcal E$ defined on the domain $[0, L_{ij}] \times [0,T]$.  Each pipe $(i,j)$ is characterized by its length $L_{ij}$, diameter $D_{ij}$, and friction factor $\lambda_{ij}$.  In addition, the cross-sectional area of each pipe is denoted by $A_{ij}$. Between any two junctions that are connected via a pipe, the mass flux and density evolve according to Eq. \eqref{eq:pde_nondim}.  Hence for each edge $(i,j)\in \mathcal E$, the evolution of $\rho_{ij}$ and $\varphi_{ij}$ is given by Eq. \eqref{eq:pde_nondim}, \emph{i.e.},
\begin{subequations}
\begin{flalign}
& \partial_t \rho_{ij} + \partial_x \varphi_{ij} = 0 \,\, \text{and} \,\, \partial_x \rho_{ij} = -\frac{\lambda_{ij} \ell_0}{2D_{ij}} \frac{\varphi_{ij} |\varphi_{ij}|}{\rho_{ij}} & \label{eq:pde_nondim_edge} \\
& \text{or, equivalently} \nonumber \\ \,\, & \partial_t \rho_{ij} + \partial_x \varphi_{ij} = 0 \,\, \text{and} \,\, \varphi_{ij} + f_{ij}(t, \rho_{ij}, \partial_x \rho_{ij}) = 0, & \label{eq:pde_nondim_edge_f}
\end{flalign}
\label{eq:pde_nondim_edge_both}
\end{subequations}
\noindent where the dissipation function for each edge $(i,j) \in \mathcal E$ is given by 
\begin{flalign}
& f_{ij}(t, u, w) = \operatorname{sgn}(w) \sqrt{\left| - u\cdot\frac{2D_{ij}}{\lambda_{ij} \ell_0}\cdot w \right|}. & \label{eq:fij}
\end{flalign}
We remark that the $\varphi_{ij}$ is directional and positive value for $\varphi_{ij}$ denotes positive flow direction. We use a directed graph in order to denote for each edge a positive flow direction, which leads to the identity $\varphi_{ij}(x_{ij},t)=-\varphi_{ji}(L_{ij}-x_{ij},t)$. In addition, every junction $i \in \mathcal V$ is associated with a time-dependent nodal density $\rho^N_i(t): [0,T] \to \mathbb R_+$. The set of controllers is denoted by $\mathcal C\subset \mathcal E\times\{+,-\}$, where $(i,j)\equiv(i,j,+)\in\mathcal C$ denotes a controller located at node $i\in\mathcal V$ that augments the density of gas flowing into edge $(i,j)\in\mathcal E$ in the direction $i \to j$, while $(j,i)\equiv(i,j,-)\in\mathcal C$ denotes a controller located at node $j\in\mathcal V$ that augments density into edge $(i,j)\in\mathcal E$ in the direction $j\to i$. Compression is then modeled as a multiplicative ratio $\ubar{\alpha}_{ij}:[0,T]\to\mathbb R_+$ for $(i,j, +)\in\mathcal C$ and $\bar{\alpha}_{ij}:[0,T]\to\mathbb R_+$ for $(i,j, -)\in\mathcal C$.

Let $\mathcal V_s \subset \mathcal V$ denote the set of supply junctions where gas enters the network. Let $s_j(t)$ be the time-varying supply density at the junction $j \in \mathcal V_s$. Mass flux withdrawals at the remaining junctions $j \in \mathcal V_d = \mathcal V \setminus \mathcal V_s$ are denoted by $d_j(t)$. For ease of exposition, we shall refer to the $\mathcal V_s$ and $\mathcal V_d$ as the set of ``slack'' and ``non-slack'' nodes, respectively. 

We shall now establish the nodal balance equations that characterize the boundary conditions for the dynamics in Eq. \eqref{eq:pde_nondim_edge_both}. To that end, we define the densities and flows at edge domain boundaries by
\begin{subequations}
\begin{flalign}
& \ubar{\rho}_{ij}(t) \triangleq \rho_{ij}(t, 0), \quad \bar{\rho}_{ij}(t) \triangleq \rho_{ij}(t, L_{ij}), & \label{eq:rhobar} \\
& \ubar{\varphi}_{ij}(t) \triangleq \varphi_{ij}(t, 0), \quad \bar{\varphi}_{ij}(t) \triangleq \varphi_{ij}(t, L_{ij}), & \label{eq:phibar}\\
& \text{and the nominal average edge flow as} \nonumber \\
& \Phi_{ij}(t) \triangleq \frac 12 (\ubar{\varphi}_{ij}(t) + \bar{\varphi}_{ij}(t))). &\label{eq:avg_phi}
\end{flalign}
\label{eq:bar_defn}
\end{subequations}
For ease of understanding, the above definitions are illustrated using a schematic in Fig. \ref{fig:schematic}, on a pipe joining two nodes $i$ and $j$.
\begin{figure}[htbp]
\centering
\includegraphics[scale=1]{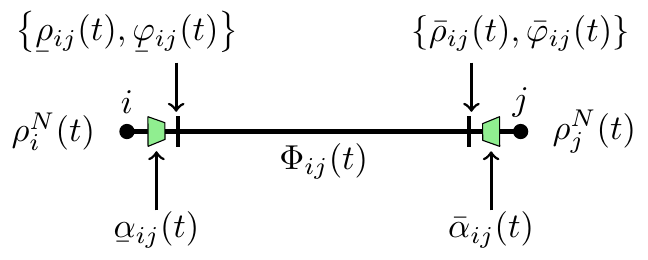}
%\begin{tikzpicture}
%    \draw[very thick] (0,0) -- (4,0);
%    \path[draw, fill=black] (0,0) circle[radius=2pt] node [left=2 mm] {$\rho_i^N(t)$};
%    \path[draw, fill=black] (4,0) circle[radius=2pt] node [right=2 mm] {$\rho_j^N(t)$};
%    \node (i) at (0,0.3) {$i$};
%    \node (j) at (4,0.3) {$j$};
%    \node (c1) at (0.3,0) {\includegraphics[scale=0.8]{compressor.pdf}};
%    \node (c2) at (3.7,0) {\includegraphics[scale=0.8,angle=180,origin=c]{compressor.pdf}};
%    \node at (0.3,-1) {$\ubar{\alpha}_{ij}(t)$};
%    \draw[thick, ->] (0.3,-0.8) -- (0.3,-0.2); 
%    \node at (3.7,-1) {$\bar{\alpha}_{ij}(t)$};
%    \draw[thick, ->] (3.7,-0.8) -- (3.7,-0.2); 
%    \draw[|-|,very thick] (0.5,0) -- (3.5,0);
%    \node at (0,1) {$\left\{\ubar{\rho}_{ij}(t), \ubar{\varphi}_{ij}(t)\right\}$};
%    \node at (4,1) {$\left\{\bar{\rho}_{ij}(t), \bar{\varphi}_{ij}(t)\right\}$};
%    \draw[thick, ->] (0.5,0.7) -- (0.5,0.2); 
%    \draw[thick, ->] (3.5,0.7) -- (3.5,0.2);
%    \node at (2,-0.3) {$\Phi_{ij}(t)$};
%\end{tikzpicture}
\caption{The figure illustrates the densities and flows at the boundaries of each edge and the compression that can be applied at both the nodes $i$ and $j$.}
\label{fig:schematic}
\end{figure}
The nodal balance equations are specified in terms of the time-dependent compressor ratios $\ubar{\alpha}_{ij}(t)$ and $\bar{\alpha}_{ij}(t)$ , the gas withdrawals $d_j(t)$, and the supply densities $s_j(t)$ as 
\begin{subequations}
\begin{flalign}
& \ubar{\rho}_{ij}(t) = \ubar{\alpha}_{ij}(t)\rho_i^N(t), \, \forall \, (i,j) \in \mathcal E, & \label{eq:nodal_density_balance_1}\\
& \bar{\rho}_{ij}(t) = \bar{\alpha}_{ij}(t)\rho_i^N(t), \, \forall \, (i,j) \in \mathcal E, & \label{eq:nodal_density_balance_2}\\
& d_j(t) =\sum_{i\in\mathcal V_d}A_{ij} \bar{\varphi}_{ij}(t)- \sum_{k\in\mathcal V_d}A_{jk}\ubar{\varphi}_{jk}(t), \,\forall\, j\in\mathcal V_d, & \label{eq:flow_balance} \\
& \ubar{\rho}_{ij}(t)=s_i(t), \,\forall\, i\in\mathcal V_s. & \label{eq:slack_pressure}
\end{flalign}
\label{eq:nodal_balance}
\end{subequations}
In Eq. \eqref{eq:flow_balance}, $A_{ij}$ denotes the cross sectional area of the pipeline $(i,j) \in \mathcal E$. 

In this study, we assume that the estimation problem is solved for a well-instrumented pipeline system for which extensive measurements are available at custody transfer meters and compressor stations.  Compressor stations are in general large facilities that have sophisticated control systems and precise measurements of pressure, temperature, flow at suction (inlet, upstream) and discharge (outlet, downstream) headers.  Because we are assuming that these measurements are available and quite accurate, we suppose that measurements of the time-varying compressor ratio functions $\{\ubar{\alpha}_{ij}, \bar{\alpha}_{ij}\}_{(i,j)\in\mathcal C}$ and the transient withdrawals $\{d_j\}_{j\in\mathcal V_d}$ are available a priori. Optimal control problems aimed at computing these compression functions given transient withdrawals and algorithms to solve the same under transient conditions have been previously addressed in the literature \cite{Zlotnik2015}. For the purpose of estimation the compression functions are assumed to be known and possibly noisy or uncertain measurements of the withdrawal functions are assumed to be available. Furthermore, the time-varying pressure at the slack nodes are also assumed to be known, because the control policy for this quantity may be determined a priori. 

Next, we suppose that densities of gas in a pipe are upper and lower bounded according to the constraints
\begin{flalign}
& \rho_{ij}^{\min}\leqslant \rho_{ij}(t,x_{ij}) \leqslant \rho_{ij}^{\max}, \,\forall\, (i,j)\in\mathcal E. &\label{eq:box_density}
\end{flalign}
For the time horizon $T$, the value of interest is $T=24$ hours during which the significant transients of interest occur. During this time horizon, we require that the state variables $\varphi_{ij}$ and $\rho_{ij}$ are time-periodic in order for the dynamic constraints to be well-posed.  To achieve time-periodicity on the variables $\varphi_{ij}$ and $\rho_{ij}$, time-periodicity also has to be imposed on the control and parameter functions $\{\ubar{\alpha}_{ij}, \bar{\alpha}_{ij}\}_{(i,j)\in\mathcal C}$, $\{d_j\}_{j\in\mathcal V_d}$, and $\{s_j\}_{j\in\mathcal V_s}$ as given by Eqs. \eqref{eq:termcon1c}--\eqref{eq:termcon1e} (see \cite{Zlotnik2015}).
That is, we impose the terminal conditions 
\begin{subequations}
\begin{flalign}
&\rho_{ij}(0,x_{ij})=\rho_{ij}(T,x_{ij}),  \,\forall\, (i,j)\in\mathcal E, &\label{eq:termcon1a}\\
&\phi_{ij}(0,x_{ij})=\phi_{ij}(T,x_{ij}),  \,\forall\, (i,j)\in\mathcal E, &\label{eq:termcon1b}\\
&\ubar{\alpha}_{ij}(0)=\ubar{\alpha}_{ij}(T), \, \bar{\alpha}_{ij}(0)=\bar{\alpha}_{ij}(T),  \,\forall\, (i,j)\in\mathcal C, &\label{eq:termcon1c} \\
&d_j(0)= d_j(T), \,\forall\, j \in \mathcal V_d, &\label{eq:termcon1d} \\
&s_j(0)=s_j(T), \,\forall\, j \in \mathcal V_s. &\label{eq:termcon1e}
\end{flalign}
\label{eq:termcon}
\end{subequations}
On one hand, this is a reasonable assumption because operators of gas pipeline systems mandate the system to be restored to a nominal state at the end of every day, which in practice is relaxed to restoration of line-pack (the total mass of gas) within local subsystems of the pipeline.  However, the primary justification for solving a problem formulated with time-periodic boundary conditions is the need for well-posedness, both conceptually and computationally. Conceptually, without some specification of the initial and terminal conditions, these states could be produced by the solver in unpredictable ways.  However, there are no obvious criteria for what these endpoint states should be for optimal control and estimation problems for pipeline networks.  The requirement of time periodicity places the formulation within a well-understood structure so that the infinite-dimensional states of the underlying PDE-constrained problem are forced to be smooth, time-periodic manifolds.  Computationally, time-periodicity can be implemented by using a single vector to store the initial and terminal points in the discretization.  This ``circular'' time-discretization reduces the problem size and eliminates the need for additional constraints on the initial and terminal states.

Furthermore, for assimilation of data that is not time-periodic, the estimation problem can be applied to an extended time-horizon over which the data are interpolated to produce periodic inputs, and the solution can be taken as the restriction to the time-horizon of interest.  We focus here on modeling and basic formulations, and leave explicit treatment of estimation using non-periodic data to future work.

% though if the initial conditions are specified we may use
% \begin{flalign}
% & \sum_{(i,j)\in\mathcal E}\int_0^{L_{ij}}\rho_{ij}(0,x)\,dx = \sum_{(i,j)\in\mathcal E}\int_0^{L_{ij}}\rho_{ij}(T,x)\,d x, & \label{termcon2}
% \end{flalign}
% which is a periodic condition on total mass in the system.

\subsection{Control system model} \label{subsec:cs_model}
In this section, we develop a reduced order model that represents the dynamics of gas flow through a network using a synthesis of Eqs. \eqref{eq:pde_nondim_edge_both}, \eqref{eq:nodal_balance}, \eqref{eq:box_density}, and \eqref{eq:termcon}. Specifically, we create a control system model using a lumped element approximation to characterize the dynamics for each edge in Eq. \eqref{eq:pde_nondim_edge_both}, together with Eq. \eqref{eq:nodal_balance} and Eq. \eqref{eq:box_density}, that uses nodal density $\rho_i^N$ for every $i \in \mathcal V$ as the state of the system. This reduction extends the previous modeling work  \cite{Grundel2013,Zlotnik2015,Zlotnik2016}. To that end, we shall first introduce a few definitions used in \cite{Zlotnik2016}.  
\begin{definition} \textit{Spatial Graph Refinement}: A refinement $\hat{\mathcal G} = (\hat{\mathcal V}, \hat{\mathcal E})$ of a directed graph $\mathcal G = (\mathcal V, \mathcal E)$ with a length $L_{ij}$ associated with each edge $(i,j) \in \mathcal E$ is constructed by adding extra nodes to subdivide the edges of $\mathcal E$ such that, the length of a new edge $(i,j) \in \hat{\mathcal E}$, $\hat L_{ij}$ satisfies 
\begin{flalign}
& \frac{\Delta L_{\mu(ij)}}{\Delta + L_{\mu(ij)}} < \hat L_{ij} < \Delta, & \label{eq:refinement} 
\end{flalign}
where, $\mu:\hat{\mathcal E} \to \mathcal E$ is a surjection from the refined edges to the parent edges and $\Delta$ denotes the maximum edge length in the refined graph produced by spatial discretization. 
\end{definition}
Throughout the rest of the article, we shall not show the explicit dependence of the density and flow variables on the independent variable $t$.
When refining the graph $\mathcal G$, we assume that $L = \Delta$ is small enough that the relative difference of the density and mass flux at the start and end of this refined edge is small \emph{i.e.}, 
\begin{flalign}
& \frac{\bar{\rho}_{ij} - \ubar{\rho}_{ij}}{\bar{\rho}_{ij} + \ubar{\rho}_{ij}} \ll 1, \quad\text{and}\quad \frac{\bar{\varphi}_{ij} - \ubar{\varphi}_{ij}}{\bar{\varphi}_{ij} + \ubar{\varphi}_{ij}} \ll 1 \, \forall \, (i,j) \in \hat{\mathcal E}. & \label{eq:lumping_condition}
\end{flalign}
The assumption indicates that $L$ is sufficiently small so that the relative density difference between the neighboring nodes is minor (as defined in Eq. \eqref{eq:lumping_condition}) at all times. The dynamics for each pipeline segment $(i,j) \in \hat{\mathcal E}$ in the refined graph $\hat{\mathcal G} = (\hat{\mathcal V}, \hat{\mathcal E})$ is again given by \eqref{eq:pde_nondim_edge_both}.
\begin{subequations}
\begin{flalign}
& \int_0^L(\partial_t\rho_{ij}+\partial_x\varphi_{ij}) \,dx=0, & \label{eq:mass_int} \\
& \int_0^L(\partial_x \rho_{ij})\,dx = -\frac{\lambda_{ij}\ell_0}{2D_{ij}}\int_0^L \frac{\varphi_{ij} |\varphi_{ij}|}{\rho_{ij}}\,dx. & \label{eq:momentum_int}
\end{flalign}
\label{eq:pde_int}
\end{subequations}
The above integrals of $\partial_t$, $\partial_x$, and nonlinear terms are evaluated using the trapezoid rule, the fundamental theorem of calculus, and averaging variables, respectively. This approximation yields 
\begin{subequations}
\begin{flalign}
& \frac{L}{2}(\dot{\ubar{\rho}}_{ij}+\dot{\bar{\rho}}_{ij}) = \ubar{\varphi}_{ij}-\bar{\varphi}_{ij}, & \label{eq:mass_disc}\\
& \ubar{\rho}_{ij}-\bar{\rho}_{ij} = -\frac{\lambda_{ij}\ell_0 L}{4D_{ij}} \frac{(\ubar{\varphi}_{ij}+\bar{\varphi}_{ij}) |\ubar{\varphi}_{ij}+\bar{\varphi}_{ij}|}{\ubar{\rho}_{ij}+\bar{\rho}_{ij}}. & \label{eq:momentum_disc}
\end{flalign}
\label{eq:ode_disc}
\end{subequations}
We remark that Eq. \eqref{eq:momentum_disc} can equivalently be written using the dissipation function as 
\begin{flalign}
& \Phi_{ij} + f_{\mu(ij)}\left(t, \frac{\ubar{\rho}_{ij} + \bar{\rho}_{ij}}{2}, \frac{\ubar{\rho}_{ij} - \bar{\rho}_{ij}}L\right) = 0 & \label{eq:momentum_f}
\end{flalign}
where $\Phi_{ij} = \frac 12 (\ubar{\varphi}_{ij} + \bar{\varphi}_{ij})$.

The resulting equations \eqref{eq:ode_disc} and nodal balance equations \eqref{eq:nodal_balance} then reduce to a differential-algebraic equation (DAE) system:
\begin{subequations}
\begin{flalign}
& \frac{L}{2}(\dot{\ubar{\rho}}_{ij}+\dot{\bar{\rho}}_{ij}) = \ubar{\varphi}_{ij}-\bar{\varphi}_{ij},\, \forall \, (i,j) \in \hat{\mathcal E}& \label{eq:dae0c}\\
& \ubar{\rho}_{ij}-\bar{\rho}_{ij} = -\frac{\lambda_{ij}\ell_0 L}{4D_{ij}} \frac{\Phi_{ij}|\Phi_{ij}|}{(\ubar{\rho}_{ij}+\bar{\rho}_{ij})},  \, \forall \, (i,j) \in \hat{\mathcal E} & \label{eq:dae0d} \\
& \ubar{\rho}_{ij} = \ubar{\alpha}_{ij}\rho_i^N, \,\bar{\rho}_{ij} = \bar{\alpha}_{ij}\rho_i^N, \, \forall \, (i,j) \in \hat{\mathcal E}, &\label{eq:dae0a}\\
& d_j = \sum_{i\in\hat{\mathcal V}_d}A_{ij} \bar{\varphi}_{ij} - \sum_{k\in\hat{\mathcal V}_d}A_{jk}\ubar{\varphi}_{jk}, \, \forall \, j\in\hat{\mathcal V}_d,  \label{eq:dae0b} & \\
& \ubar{\rho}_{ij} = s_i, \,\forall \, i\in\hat{\mathcal V}_s. & \label{eq:dae0e}
\end{flalign}
\label{eq:dae}
\end{subequations}
Eq. \eqref{eq:dae0a} represents continuity of density at junctions with jumps in the case of compression or regulation, Eq. \eqref{eq:dae0b} represents flow balance at junctions, and Eqs. \eqref{eq:dae0c}-\eqref{eq:dae0d} represent flow dynamics on each segment.

The DAE system in Eq. \eqref{eq:dae} can be equivalently represented by a system of nonlinear DAEs in matrix-vector form using graph theoretic notation. We shall first introduce the additional graph-theoretic notation and state the resulting DAE system, while providing the derivation in an Appendix. The set of nodes in the set $\hat{\mathcal V}$ is first enumerated according to a fixed ordering. For ease of exposition, we choose an ordering where the non-slack nodes, $\hat{\mathcal V}_d$, are ordered after the slack nodes, $\hat{\mathcal V}_s$. Now each node in $\hat{\mathcal V}$ is assigned an index $[\hat{\mathcal V}] := \{1,\dots, |\hat{\mathcal V}|\}$ according to the chosen ordering. Each edge is also assigned an index in $[\hat{\mathcal E}] := \{1, \dots, |\hat{\mathcal E}|\}$ and we define the map $\pi_e:\hat{\mathcal E} \to [\hat{\mathcal E}]$, which maps each edge to this ordering. Throughout the rest of the article, boldface notation will be used to represent vectors.

Let $\bm \rho^N = (\rho_1^N, \rho_2^N, \dots, \rho_{|\hat{\mathcal V}|}^N)^{\intercal}$ denote the nodal density state vector. 
Equation \eqref{eq:dae0a} will be used to state \eqref{eq:dae0c}-\eqref{eq:dae0d} in terms of nodal densities $\bm \rho^N$.  We then define state vectors $\ubar{\bm \varphi}=(\ubar{\varphi}_1,\ldots,\ubar{\varphi}_{|\hat{\mathcal E}|})^\intercal$ and $\bar{\bm \varphi}=(\bar{\varphi}_1,\ldots,\bar{\varphi}_{|\hat{\mathcal E}|})^\intercal$, where $\ubar{\varphi}_k$ and $\bar{\varphi}_k$ are indexed by $k=\pi_e(ij)$. Furthermore, we let $\bm \Phi = \tfrac 12 (\ubar{\bm \varphi} + \bar{\bm \varphi})$ denote the vector of average flow in each pipeline segment.

We now define the incidence matrix of the full refined graph $\hat{\mathcal G}$, acting
$A:\mathbb R^{|\hat{\mathcal E}|}\to\mathbb R^{|\hat{\mathcal V}|}$, by
\begin{flalign} \label{eq:incidence0}
&A_{ik} = \left\{ \begin{array}{ll}  1 & \text{edge $k=\pi_e(ij)$ enters node $i$,} \\ -1 & \text{edge $k=\pi_e(ij)$ leaves node $i$,} \\ 0 & \text{else} \end{array}\right. &
\end{flalign}
We then define the time-dependent weighted incidence matrix $B:\mathbb R^{|\hat{\mathcal E}|}\to\mathbb R^{|\hat{\mathcal V}|}$ by
\begin{flalign} \label{eq:incidence0_w}
&B_{ik} = \left\{ \begin{array}{ll}  \bar{\alpha}_{ij} & \text{edge $k=\pi_e(ij)$ enters node $i$,} \\ -\ubar{\alpha}_{ij} & \text{edge $k=\pi_e(ij)$ leaves node $i$,} \\ 0 & \text{else}. \end{array}\right. &
\end{flalign}
where $\operatorname{sign}(B)=A$.  In the reduced order model, the compressor control inputs are embedded within the matrix $B$. We define the vector of withdrawal fluxes $\bm d=(d_1,\ldots,d_M)^T$ with $M=|\hat{\mathcal V}_d|$, where $d_k$ is negative if an injection. Also define the slack node densities as $\bm s=(s_1,\ldots,s_b)^{\intercal}=\{\rho^N_j\}_{j\in\hat{\mathcal V}_s}$, where $b=|\hat{\mathcal V}_s|$, and non-slack (demand) node densities as $\bm \rho=(\rho_1,\ldots,\rho_M)^T=\{\rho^N_j\}_{j\in\hat{\mathcal V}_d}$, so that $b+M=|\hat{\mathcal V}|$. We remark that $\bm s$, $\bm \rho$ and $\bm \rho^N$ are related as $\bm \rho^N = (\bm s, \bm \rho)^\intercal$, because of the chosen ordering of nodes in $\hat{\mathcal V}$. Then let $A_s,B_s\in\mathbb R^{b\times |\hat{\mathcal E}|}$ denote the sub-matrices of rows of $A$ and $B$ corresponding to $\hat{\mathcal V}_s$, and let $A_d,B_d\in\mathbb R^{M\times |\hat{\mathcal E}|}$ correspond similarly to $\hat{\mathcal V}_d$.  
%Next, let $A_L$ and $A_0$ denote the positive and negative parts of $A_d$, so that $A_d=A_L+A_0$.  
Also, define the diagonal matrices $\Lambda,K, X\in\mathbb R^{|\hat{\mathcal E}|\times |\hat{\mathcal E}|}$ by $\Lambda_{kk}=L_k$,  $K_{kk}=\ell_0\lambda_k/D_k$, and $X_{kk} = A_k$ where $L_k$, $\lambda_k$, $D_k$, and $A_k$ are the non-dimensional length, friction factor, diameter, and cross-sectional area of edge $k=\pi_e(ij)$. Then \eqref{eq:dae} can be rewritten (see Appendix \ref{app:dae} for proof) as
\begin{subequations}
\begin{flalign}
& |A_d| X \Lambda |B_d^\intercal|\dot{\bm \rho} = 4(A_d X \bm \Phi - \bm d) - |A_d| X \Lambda |B_s^\intercal| \dot{\bm s}  & \label{eq:dae1a} \\
& \Lambda K \bm \Phi \odot \bm \Phi =  -B^\intercal \bm \rho^N \odot |B^\intercal| \bm \rho^N & \label{eq:dae1b} 
\end{flalign}
\label{eq:dae_final}
\end{subequations}
where the operator $\odot$ represents the Hadamard product. Here, the gas withdrawals are $\bm d \in \mathbb R^M$, input densities are $\bm s \in \mathbb R_+^b$ and the compression ratios $\ubar{\alpha}_{ij}, \bar{\alpha}_{ij} \in \mathcal C$ are time-varying and $\bm \rho \in \mathbb R_+^M$ and $\bm \Phi \in \mathbb R^{|\hat{\mathcal E}|}$ denote the states of the system. The system of equations in \eqref{eq:dae_final} is a DAE system. This DAE system can be converted to a system of ODEs with nodal densities $\bm \rho$ as the only set of state variables by expressing each term in $\bm \Phi$ in terms of $\bm \rho$ using Eq. \eqref{eq:dae1b} and substituting them in Eq. \eqref{eq:dae1a}. The vector $\bm \Phi$ expressed in terms of $\bm \rho$ is given by
\begin{flalign}
& \bm \Phi = -  \left| - (\Lambda K)^{-1} B^\intercal \bm \rho^N \odot |B^\intercal| \bm \rho^N \right|^{\frac 12} \odot \operatorname{sgn}(B^\intercal \bm \rho^N) & \label{eq:phi_vec_f} 
\end{flalign}
where the signum function is being applied component-wise to the vector $B^\intercal \bm \rho^N$. Alternatively, the $k^{\text{th}}$ component of $\bm \Phi$ in Eq. \eqref{eq:phi_vec_f} can be expressed in terms of the dissipation function in Eq. \eqref{eq:fij} as follows:
\begin{flalign}
& \Phi_k = - f_{\mu(ij)}\left(t, \frac 12 \left( |B^\intercal| \bm \rho^N \right)_k, \left( \Lambda^{-1} B^\intercal \bm \rho^N\right)_k\right), & \label{eq:phi_f}
\end{flalign}
where, $k = \pi_e(ij)$ and $(i,j) \in \hat{\mathcal E}$. In the above equation, $\left(\cdot\right)_k$ denotes the $k^{\text{th}}$ component of the vector. For ease of exposition, we will use $\bm \Phi + \bm f(\cdot,\cdot,\cdot) = 0$ to equivalently represent \eqref{eq:phi_vec_f}. Substituting for $\bm \Phi$ in Eq. \eqref{eq:dae1a} using Eq. \eqref{eq:phi_vec_f}, we obtain a set of nonlinear ODEs that represent flow of gas through the network in terms of purely the nodal density dynamics as
\begin{flalign}
& |A_d| X \Lambda \left( |B_d^\intercal|\dot{\bm \rho} + |B_s^\intercal| \dot{\bm s} \right) + 4(A_d X \bm f(\cdot,\cdot,\cdot) + \bm d) = \bm 0.  & \label{eq:ode_nodal}
\end{flalign}

\subsection{Uncertainty modeling} \label{subsec:noise}
Given the nodal density dynamics in Eq. \eqref{eq:ode_nodal}, we incorporate an additive noise process $\bm \eta$ as given below:
\begin{flalign}
& |A_d| X \Lambda \left( |B_d^\intercal|\dot{\tilde{\bm \rho}} + |B_s^\intercal| \dot{\bm s} \right) + 4(A_d X \bm f(\cdot,\cdot,\cdot) + \tilde{\bm d}) + \bm \eta = \bm 0.  & \label{eq:ode_nodal_noise}
\end{flalign}
In the above equation, $\bm \eta$ is time varying and has a dependence on time that has not been made explicit for sake of readability.  We use the noise process $\bm \eta$ to simultaneously account for errors caused by (i) simplification of physical modeling, (ii) uncertainty in model parameters, and (iii) process and measurement noise.  All these uncertainties are lumped together into one additive noise process $\bm \eta$.  Specifically, the value $\tilde{\bm \rho}$ represents the solution to the stochastic DAE in \eqref{eq:ode_nodal_noise} given stochastic withdrawals $\tilde{\bm d}$.  For the estimation problems that we will treat, we assume that  $\tilde{\bm d}$ and $\tilde{\bm \rho}$ are available, and we interpret them as noisy measurements of $\bm d$ and $\bm \rho$, which in turn satisfy the deterministic (noiseless) model \eqref{eq:ode_nodal}.    In addition, after time-discretization of the system \eqref{eq:ode_nodal_noise} the process $\bm \eta$ can be interpreted to incorporate error caused by coarse sampling in time.  %apart from the fact that it is an additive %white Gaussian noise process.  
Throughout the rest of the article, we shall, without loss of generality, refer to $\bm \eta$ as the measurement noise.  We do not make any other assumption on $\bm \eta$.  In Section \ref{sec:state_param}, we formulate least squares problems for weighted $L_2$ minimization of the measurement and process errors over a time horizon $T$, i.e., $\int_0^T(\bm d-\tilde{\bm d})^{\intercal} W_1(\bm d-\tilde{\bm d}) \mathrm{d}t$ and $\int_0^T(\bm \rho - \tilde{\bm \rho})^\intercal W_2(\bm \rho - \tilde{\bm \rho}) \mathrm{d}t$, with respective weighting matrices $W_1$ and $W_2$, subject to the deterministic dynamic constraints \eqref{eq:ode_nodal}.  %These formulations minimize the $L_2$ norm of $\sigma(\bm \eta)$, where $\sigma(\cdot)$ is a nonlinear transformation that depends on the system \eqref{eq:ode_nodal} and its parameters.  
Our purpose is to develop an applied technique for state and parameter estimation for pipeline system models of the form \eqref{eq:ode_nodal}, so we do not attempt to analyze the characteristics of $\bm \eta$ here. Rather, we will use empirical studies to characterize performance of the developed estimation approach.

\section{Assumptions and Monotonicity properties of the Nodal Dynamics} \label{sec:assumptions}
In this section, we state the assumptions made on the nonlinear system of ODE in Eq. \eqref{eq:ode_nodal} that represent the nodal dynamics and derive monotonicity properties for the same. 
\subsection{Assumptions on the nodal dynamics} \label{subsec:assumptions}
We formulate formal assumptions for the nodal dynamics in Eq. \eqref{eq:ode_nodal} that are necessary to establish the monotonicity property required for the proposition in the next section \cite{Misra2016}.  
\begin{enumerate}
\item \label{assump:3} Differentiability of inputs and controls: The compressor functions $\ubar{\alpha}_{ij}(t)$ and $\bar{\alpha}_{ij}(t)$, the gas withdrawal profiles $d_j(t)$, and the input densities $s_j(t)$ at the slack nodes are all $C^k([0,T])$ functions for $k\geq 2$. 
\item \label{assump:1} Well-posedness of the initial value problem: The initial value problem for the Eq. \eqref{eq:ode_nodal} along with the initial conditions $\bm \rho(0) = \bm \rho^0$ for given time-varying twice differentiable withdrawal profiles, input densities, and compression ratios has a unique solution that is twice differentiable. 
\item \label{assump:2} Existence of solution to the boundary value problem with time-periodic boundary conditions: A solution exists for the boundary value problem on the system of ODEs given by Eq. \eqref{eq:ode_nodal} with time-periodic boundary conditions $\bm \rho(0) = \bm \rho(T)$, given $\{\bar{\alpha}_{ij}, \ubar{\alpha}_{ij}\}_{(i,j) \in \mathcal C}$, $\bm d$ and $\bm s$ such that 
\begin{subequations}
\begin{flalign}
& \ubar{\alpha}_{ij}(0) = \ubar{\alpha}_{ij}(T), \, \forall \, (i,j) \in \mathcal C, & \label{eq:c_tp1} \\
& \bar{\alpha}_{ij}(0) = \bar{\alpha}_{ij}(T)\, \forall \, (i,j) \in \mathcal C & \label{eq:c_tp2} \\
& \bm d(0) = \bm d(T), \,\text{and}\, \bm s(0) = \bm s(T). & \label{eq:ds_tp}
\end{flalign}
\label{eq:tp_parameters}
\end{subequations}
\end{enumerate}
The above assumptions are imposed in order to guarantee that the flow solution defined by the original PDE system \eqref{eq:pde_1}, with consistent boundary conditions specified by the nodal balance conditions \eqref{eq:nodal_balance} on each pipe, admits a unique classical solution that is both mathematically well-defined and physically realizable.  In this case, it is straightforward to show the discretized system \eqref{eq:ode_nodal} is a consistent approximation of the PDE system \cite{Zlotnik2015}. We suppose that in the limit as the discretization step $\Delta$ approaches zero, the solution to the nodal dynamics \eqref{eq:ode_nodal} will approach the solution to the full PDE dynamics pointwise, although a rigorous proof of this result is outside the scope of the present study.  The above assumptions are not restrictive, in the sense that they are reasonable for the physical system that is being studied.  The dynamics of gas flows in the regime of slowly-varying transients that do not exhibit waves or shocks will have a twice-differentiable solution to an initial boundary value problem when the parameter functions are twice-differentiable (or can be closely approximated by twice-differentiable functions), and when the initial conditions would not induce such shocks after time $t=0$ (flows and pressures are balanced at nodes). With the above assumptions of well-posedness of the initial value problem and existence of a solution given twice-differentiable inputs, we proceed to establish a uniqueness result for the state and parameter estimation solution.

\subsection{Monotonicity of nodal dynamics} \label{subsec:monotonicity}
We first introduce a few definitions before presenting the monotonicity properties of the nodal dynamics in Eq. \eqref{eq:ode_nodal}. 
\begin{definition} \label{def:monotone_cs}
\textit{Monotone-parametrized control system \cite{Zlotnik2016}:} Let 
\begin{flalign}
\dot{\bm x} = g(\bm x, \bm u, \bm p), \qquad \bm x (0) = \bm y \label{eq:cs}
\end{flalign}
where, $\bm x(t) \in \mathbb R^n$ is the state vector, $\bm u \in \mathbb R^m$ is the control vector, and $p(t) \in \mathbb R^p$ is a parameter vector. Furthermore, $g$ is Lipschitz. The control system \eqref{eq:cs} is monotone-parameterized with respect to the parameter vector $\bm p(t)$ if, for all $i \geqslant 0$, $\bm y_1, \bm y_2 \in \mathbb R^n$, $\bm u(t): (0,\infty) \rightarrow \mathbb R^m$, and piecewise-continuous functions $\bm p_1(t), \bm p_2(t): (0,\infty) \rightarrow \mathbb R^p$, the orderings $\bm y_1 \leqslant \bm y_2$ and $\bm p_1(s) \leqslant \bm p_2(s)$ for all $s \in [0,t]$ imply that $\bm x_1(t) \leqslant  \bm x_2(t)$. Here, the inequalities for vectors are mean componentwise and $\bm x_1(t)$, $\bm x_2(t)$ are the solutions to \eqref{eq:cs} with initial condition $\bm y_1$, $\bm y_2$, control input $\bm u(t)$, and parameter vectors $\bm p_1(t)$, $\bm p_2(t)$, respectively. 
\end{definition}
Using the above definition, we subsequently show that the nodal dynamics in Eq. \eqref{eq:ode_nodal} is mononote-parameterized with respect to the withdrawals, \textit{i.e.}, the nodal density/pressure can only increase with decreasing withdrawals. In Eq. \eqref{eq:ode_nodal}, the withdrawals are time-varying parameters. Mathematical conditions for which systems representing actuated flows in dissipative flow networks are monotonic have been derived \cite{Zlotnik2016,Misra2016}. We re-state the conditions without proof.

\begin{proposition} \label{prop:monotonicity}
(see \cite{Zlotnik2016}, Proposition 2) For a system representing actuated flows in dissipative flow networks with positive and differentiable compression functions, the nodal flow dynamics obtained via spatial discretization are monotone parameterized with respect to nodal withdrawals if the dissipation function $f_{ij}(t, u, w)$ is differentiable and increasing in its last argument for all $(i,j) \in \mathcal E$ \emph{i.e.}, 
\begin{flalign}
    \frac{\partial}{\partial w} f_{ij}(t, u, w) > 0, \quad (i,j) \in \mathcal E. \label{eq:monotone_condition}
\end{flalign}
\end{proposition}

We shall now show that the nodal dynamics given by Eq. \eqref{eq:ode_nodal} is also monotone with respect to the nodal injections. We will first need the following lemma:

\begin{lemma} \label{lem1:monotonicity}
Consider a connected graph $\hat{\mathcal G}=(\mathcal{V},\mathcal{E})$ with more than two nodes that has incidence matrix $A$. Let $0\leqslant k<|\mathcal{V}|-2$. A matrix $A_d$ created by removing any $k+1$ rows from $A$ leaves $|\hat{\mathcal V}|-1-k$ linearly independent rows.
\end{lemma}
\begin{proof}
Let $r(Q)$ denote the rank of a matrix $Q$.  For a connected graph with more than 2 nodes, the incidence matrix $A$ satisfies $r(A)=|\hat{\mathcal V}|-1$ (\cite{Bapat2014}, Lemma 2.2).  By definition, each column of $A$ has two non-zero entries, which are $1$ and $-1$.  For two rows $a_j$ and $a_k$ of $A$ to be linearly dependent, they must satisfy $a_k=-a_j$, which means that two nodes of the graph are connected by one or more edges, but not connected to any other node in the graph, contradicting the assumption of connectedness.  It follows that all rows of $A$ are unique.  Because $A$ has $|\hat{\mathcal V}|$ rows, removing one row leaves $|\hat{\mathcal V}|-1$ linearly independent rows. Furthermore, removing $k$ additional rows leaves $|\hat{\mathcal V}|-1-k$ independent rows. 
\end{proof}

\begin{corollary} \label{cor:monotonicity}
Given control profiles $\ubar{\alpha}_{ij}, \bar{\alpha}_{ij} \in C_+^k([0,T])$, the nodal dynamics in Eq. \eqref{eq:ode_nodal} are monotone parameterized with respect to the nodal withdrawals, $\bm d$.
\end{corollary}
\begin{proof}
Recall that $A_d$ is obtained by removing $A_s$ from the incidence matrix $A$ of the graph $\mathcal{G}$.  If $|\hat{\mathcal V}_s|=k+1$ for some integer $k\geqslant 0$, by Lemma \ref{lem1:monotonicity} we see that $r(A_d)=|\hat{\mathcal V}_d|$, i.e., that $A_d$ is full rank.  It is straightforward to also show that $r(B_d)=r(A_d)$.  Because $|A_d|$ and $|B_d|$ are full rank and positive, and $X$ and $\Lambda$ are diagonal and positive, $|A_d| X \Lambda |B_d^\intercal|$ is invertible.

By multiplying Eq. \eqref{eq:ode_nodal} by $(|A_d| X \Lambda |B_d^\intercal|)^{-1}$ and re-arranging terms, the nodal dynamics can be written as
\begin{flalign}
\dot{\bm \rho} & = (|A_d| X \Lambda |B_d^\intercal|)^{-1} \nonumber \\ & \qquad \times [4(A_d X \bm \Phi - \bm d) - |A_d| X \Lambda |B_s^\intercal| \dot{\bm s}],  & \label{eq:ode_nodal_mod}
\end{flalign}
where the $k$th entry of $\bm \Phi$ is given by
\begin{flalign}
& \Phi_k = - f_{\mu(ij)}\left(t, \frac 12 \left( |B^\intercal| \bm \rho^N \right)_k, \left( \Lambda^{-1} B^\intercal \bm \rho^N\right)_k\right). & \label{eq:phi_f_mod}
\end{flalign}
The discretized dynamics as expressed in \eqref{eq:ode_nodal_mod}-\eqref{eq:phi_f_mod} are of the ODE form of the monotone parameterized control system in Definition \ref{def:monotone_cs}, and therefore Prop. \ref{prop:monotonicity} may be applied.

From Prop. \ref{prop:monotonicity}, it is sufficient to show that Eq. \eqref{eq:monotone_condition} is satisfied for the dissipation function in Eq. \eqref{eq:fij}. To that end, 
\begin{flalign*}
& \frac{\partial}{\partial w}f_{ij}(t, u, w) = \operatorname{sgn}(w) \cdot \frac 12 \left| - u\cdot\frac{2D_{ij}}{\lambda_{ij} \ell_0}\cdot w \right|^{-\tfrac 12} \cdot \operatorname{sgn}(-u\cdot w) \cdot & \\ 
& \qquad \qquad \qquad \qquad \qquad \qquad \qquad \qquad \left(-u\cdot \frac{2D_{ij}}{\lambda_{ij} \ell_0}\right) &  \\
& = \operatorname{sgn}(w^2) \cdot \frac 12 \left| - u\cdot\frac{2D_{ij}}{\lambda_{ij} \ell_0}\cdot w \right|^{-\tfrac 12} \cdot \left(u\cdot \frac{2D_{ij}}{\lambda_{ij} \ell_0}\right) & \\
& = \frac 12 \left| - u\cdot\frac{2D_{ij}}{\lambda_{ij} \ell_0}\cdot w \right|^{-\tfrac 12} \cdot \left(u\cdot \frac{2D_{ij}}{\lambda_{ij} \ell_0}\right)
\end{flalign*}
Because  $\frac 12 \left( |B^\intercal| \bm \rho^N \right)_k > 0$ in \eqref{eq:phi_f_mod}, it follows that $u > 0$ holds for each evaluation of $f$, so that the final expression above is positive.  Therefore $\frac{\partial}{\partial w}f_{ij}(t, u, w) > 0$, which completes the proof. 
\end{proof}
While it has been shown previously that discretizations of dissipative flow networks possess the monotonicity property \cite{Zlotnik2016}, the result above proves that the same property holds for the specific model of natural gas network transients derived here, and in particular for the specific discretization employed.  Note that although the DAE system \eqref{eq:ode_nodal} may be written in the ODE form \eqref{eq:ode_nodal_mod}-\eqref{eq:phi_f_mod} for analytical convenience, we use the nodal DAE system in practice to avoid numerical ill-conditioning. Corollary \ref{cor:monotonicity} is a powerful result that will be used to establish a uniqueness property for time-periodic boundary value problems for the nodal dynamics \eqref{eq:ode_nodal}, which yields important implications for state observability.

\section{Estimation Problem Formulations} \label{sec:state_param}
We now present the formulations for the state estimation and the joint state  and parameter estimation problems. All the estimation problems are formulated as nonlinear least squares problem when the time-varying withdrawals and the compressor ratios are known a priori. We present a formulation for each of the following three problems in that order: (1) state estimation with noiseless withdrawal values, (2) state estimation with noisy (uncertain) withdrawal and nodal density measurements, and (3) joint state  and parameter estimation with noisy (uncertain) withdrawal and nodal density measurements. We also remark that all the measurements are assumed to be obtained only from the physical nodes, $\mathcal V$, in the graph $\mathcal G$.

\subsection{State estimation with noiseless withdrawal measurements} \label{subsec:state_1}
If the exact (noiseless) withdrawal profiles $\bm d$ and slack node pressures $\bm s$ are known a priori, we claim  that no additional measurements of the nodal densities or mass flux values at the edges are required to estimate all the states of the system in the presence of time-periodic boundary conditions and time-periodic compression. The claim follows from the following proposition
\begin{proposition} \label{prop:uniqueness}
Under the assumptions \ref{subsec:assumptions}, the system of ODEs in \eqref{eq:ode_nodal} admits a unique time-periodic solution for given time-periodic withdrawal profiles $\bm d$, slack node pressures $\bm s$, and compression ratios $\{ \ubar{\bm \alpha}, \bar{\bm \alpha}\}$.
\end{proposition}
\begin{proof}
We utilize the monotonicity properties of Eq. \eqref{eq:ode_nodal} to prove the result. Consider nominal profiles for the vector of slack node densities is $\bm s = \check{\bm s}$ and the vector of withdrawals $\bm d = \check{\bm d}$. We claim that there is a unique vector of non-slack nodal densities $\check{\bm \rho}$ that satisfies the system of ODEs in \eqref{eq:ode_nodal}. For the sake of contradiction, we suppose not, and let $\check{\bm \rho}_1$ and $\check{\bm \rho}_2$ denote any two distinct solutions to the ODE system. Let $\check{\bm d}_1$ and $\check{\bm d}_2$ denote the corresponding withdrawal profiles that induce these solutions. We know, by the assumption $\bm d = \check{\bm d}$, that $\check{\bm d}_1 = \check{\bm d}_2$. Hence, we have $\check{\bm d}_1 \geqslant \check{\bm d}_2$ and $\check{\bm d}_1 \leqslant \check{\bm d}_2$.  By Corollary $\ref{cor:monotonicity}$, these inequalities imply that $\check{\bm \rho}_1 \leqslant \check{\bm \rho}_2$ and $\check{\bm \rho}_1 \geqslant \check{\bm \rho}_2$. It follows that $\check{\bm \rho}_1 = \check{\bm \rho}_2 = \check{\bm \rho}$. 
\end{proof}
An important implication of the above result is that when measurements of withdrawals $\bm d$, slack node pressures $\bm s$, and compression ratios $\{ \ubar{\bm \alpha}, \bar{\bm \alpha}\}$ are available, no measurements of any states of the system are required to estimate all the states of the system.  In this particular case, the state estimation problem reduces to solution of an initial value problem (IVP) for the ODE system \eqref{eq:ode_nodal}, or equivalently, the system of DAEs given by \eqref{eq:dae_final} with given initial state $\bm{\rho}(0)$. For the subsequent estimation computations, we utilize a finite difference approximation for the derivatives in the ODE/DAE system (in Eq. \eqref{eq:ode_nodal}/\eqref{eq:dae_final}) and convert the nonlinear system of ODEs/DAEs to a system of nonlinear algebraic equations. The solution to the resulting optimization problem is then computed using an interior point solver.

\subsection{State estimation with noisy (uncertain) withdrawal and nodal density measurements} \label{subsec:state_2}
In practice, the time-periodic withdrawal profiles are measured using flow meters and are noisy or uncertain. Additionally, the pressure gauges at the junctions provide noisy nodal pressure measurements that can be converted to noisy density measurements. In this setting, the approach presented in Sec. \ref{subsec:state_1} needs to be modified to suit this case. To that end, let $\tilde{d}_j(t)$, $j \in \mathcal V_D$ denote these measured withdrawal profiles and $\tilde{\rho}_j(t)$, $j \in \mathcal V_D$ denote the density (pressure) measurements at the non-slack nodes; let $\tilde{\bm d}$ and $\tilde{\bm \rho}$ represent the corresponding vector of measurements. Similar to what has been done for the previous state estimation problem, we assume that the slack nodal pressure vector and the compression ratios are known. We can then formulate a weighted nonlinear least squares problem, where $W_1$ and $W_2$ are the weighting matrices, using the following running cost objective function: 

\begin{flalign}
\mathcal{L}(\bm d,\tilde{\bm d}, \bm \rho,\tilde{ \bm \rho}) \equiv & \int_0^T(\bm d - \tilde{\bm d})^\intercal W_1 (\bm d - \tilde{\bm d}) \nonumber \\  & \qquad +  (\bm \rho - \tilde{ \bm \rho})^\intercal W_2 (\bm \rho - \tilde{ \bm \rho}) \mathrm{d}t & \label{eq:state_2_obj} 
\end{flalign}

Then, the state estimation problem is formulated as 
\begin{subequations}
\begin{flalign}
    & \min_{\bm \rho,\,\bm\Phi,\,\bm d} \, \mathcal{L}(\bm d,\tilde{\bm d}, \bm \rho,\tilde{ \bm \rho})& \notag \\ % \text{ in Eq.} \eqref{eq:state_2_obj} & \\
    &\text{subject to: } \text{Eqs. \eqref{eq:dae_final}, } & \notag \\
    & {\bm \rho}^{\min} \leqslant {\bm \rho} \leqslant {\bm \rho}^{\max}, \, &%{\bm \Phi}^{\min} \leqslant {\bm \Phi} \leqslant {\bm \Phi}^{\max}, 
     \label{eq:state_2_bounds} \\
    & {\bm \rho}(0) = {\bm \rho}(T), \, {\bm \Phi}(0) = {\bm \Phi}(T), \text{ and } {\bm d}(0) = {\bm d}(T)  &\label{eq:state_2_periodicity}
\end{flalign}
\label{eq:state_2}
\end{subequations}
The optimized variables are densities $\bm \rho$, per-area mass flows $\bm \Phi$, and estimated withdrawals $\bm d$.  The constraints in Eq. \eqref{eq:state_2_bounds} impose bounds on the nodal density and edge flux state variables. Notice that the above nonlinear least squares formulation computes time-periodic estimates of the state variables and the withdrawals. This formulation is justified by the principle of convergence in the limit as the magnitude of noise decreases to zero.  That is, the state estimates approach the state estimates that are obtained in the noiseless withdrawal case in Sec. \ref{subsec:state_1} as the noise in the measurement $\tilde{d}$ is decreased.

The above formulation is then solved via a gradient descent algorithm after approximating the derivatives using finite differences and the integral using the trapezoidal rule. Although the time discretization used is necessarily coarse, which is required for tractability of the large-scale problems of interest, the use of such representations in optimal control problems has been validated in the regime of interest in previous studies, as discussed in Section \ref{sec:model}. A major implication of using this coarse discretization for the dynamic constraints is the restriction of measurements used in the objective function to the same coarse grid of collocation points.  In this setting, the discretized noise process $\bm \eta$ can be interpreted to also account for errors that arise from coarse time discretization of the dynamic constraints.
In the next section, we formulate a joint state  and parameter estimation problem, similar to \eqref{eq:state_2}, where the parameters that are of interest are the friction factors for all pipes. 

\subsection{State \& parameter estimation with noisy withdrawal and nodal density measurements} \label{subsec:state_3}
The friction factor of each pipe denoted by $\lambda_{ij}$ is contained within the matrix $K$ in \eqref{eq:dae1b}. For the joint state  and parameter estimation problem, in addition to including the variables $\bm \rho$, $\bm \Phi$, and $\bm d$ in the nonlinear least squares formulation, $K$ is also a diagonal variable matrix. Except for this difference, the formulation is similar to the nonlinear least squares formulation in given by Eq. \eqref{eq:state_2}. Then, the state estimation problem is formulated as 
\begin{subequations}
\begin{flalign}
    & \min_{\bm \rho,\,\bm\Phi,\,\bm d,\,K} \, \mathcal{L}(\bm d,\tilde{\bm d}, \bm \rho,\tilde{ \bm \rho}) \notag \\ % \text{ in Eq.} \eqref{eq:state_2_obj} \notag \\
    & \text{subject to: } \text{Eqs. \eqref{eq:dae_final}, } & \notag \\
    & {\bm \rho}^{\min} \leqslant {\bm \rho} \leqslant {\bm \rho}^{\max}, & %{\bm \Phi}^{\min} \leqslant {\bm \Phi} \leqslant {\bm \Phi}^{\max}, 
     \label{eq:state_2_bounds_sp} \\
    & {\bm \rho}(0) = {\bm \rho}(T), \, {\bm \Phi}(0) = {\bm \Phi}(T), \text{ and } {\bm d}(0) = {\bm d}(T).  &\label{eq:state_2_periodicity_sp}
\end{flalign}
\label{eq:state_2_sp}
\end{subequations}
Following a finite difference approximation of the derivatives, a similar interior point optimization algorithm is used to solve this formulation.  The optimized variables are densities $\bm \rho$, per-area mass flows $\bm \Phi$, estimated withdrawals $\bm d$, and friction factor parameters $\lambda_{ij}$ for each edge $\mathcal{E}$ in the original physical graph $\mathcal{G}$.

\section{Computational Results} \label{sec:results}
We now present the computational results for all of the algorithms presented in this paper. We consider three test cases: (1) a single pipe case, (2) a 4-node network, and (3) a 25-node network. Each formulation presented in Sec. \ref{sec:state_param} is converted into a nonlinear program (NLP) using a finite difference approximation on the derivative terms. This technique of converting a continuous time problem to a finite-dimensional constrained NLP has been widely used in the optimal control literature \cite{Ross2003}, and has been applied previously in the context of gas pipeline systems \cite{Zlotnik2015}. We use a primal-dual interior point solver, IPOPT \cite{Wachter2009}, together with Automatic Differentiation in Julia/JuMP \cite{Dunning2017}, to compute the jacobians, and solve the resulting NLPs. IPOPT is chosen due to its ability to leverage sparse linear algebra computations. In the following section, we present a complete description of the test cases. An error tolerance of $10^{-4}$ is used for all the computational experiments, which were evaluated on a $2.9$ GHz, Intel Core i5 machine with $16$ GB RAM. 

% The experimental setup is as follows: for each test case, the slack node pressure, gas withdrawals, the compression ratios, and the friction factors of each pipeline are assumed to be known. Using these values, a simulation of the system of ODEs is performed using the Eq. \eqref{eq:dae_final} (specify the simulator and other details). The simulation yields the density and flow rates at each node and edge, respectively. As and when required, noise is added to the density and gas withdrawals to perform state and parameter estimation with noisy measurements. During the joint parameter and state estimation, the friction factors are not known a priori and are estimated using these measurements. A schematic and the values used for the single pipe case are shown in the Fig. \ref{fig:one_pipe_schematic}.

\subsection{Description of the test cases} \label{subsec:description}
Three test instances were used to evaluate performance of the proposed computational estimation method.  The single pipe test case contains two nodes connected by a single pipeline of length $100$ km and cross-sectional diameter $0.5$ m. Gas is being supplied at one of the nodes (source node) at a pressure of $942.75$ psi; this slack node pressure is boosted immediately by a compressor at the source end. Similarly, gas is being withdrawn  according to the function $68.094\left(1+ 0.1\sin \frac{4\pi t}{T}\right)$ kg/s at the other node. The pressure of the gas as it flows through the pipeline is bounded between $500$ psi and $1100$ psi. The network model $\mathcal{G}$ for a single pipe is discretized by adding auxiliary nodes at an interval of 5 km to create the refined graph $\hat{\mathcal{G}}$.  The Fig. \ref{fig:sim_pressure} and \ref{fig:sim_flux} show the simulated nodal pressure and edge mass flux, respectively, for the refined (discretized) graph for the single pipe case with the aforementioned slack pressure, withdrawal profile, parameters, and a time horizon $T=24$ hours. All forward simulations were performed using an implicit Euler DAE integrator for the DAE system in Eq. \eqref{eq:dae_final}; the software Sundials \cite{Hindmarsh2005} was used to implement the simulation using adaptive time-stepping to meet a relative error tolerance of 10$^{-4}$. Noise is then added to these simulated values and used as input to the estimation problems. We show plots of the simulated and estimated profiles only for the single pipe case and resort to tables to illustrate the effectiveness of the estimation algorithms on the remaining two test cases. 
\begin{figure}[ht!]
    \centering
    \includegraphics[scale=0.5]{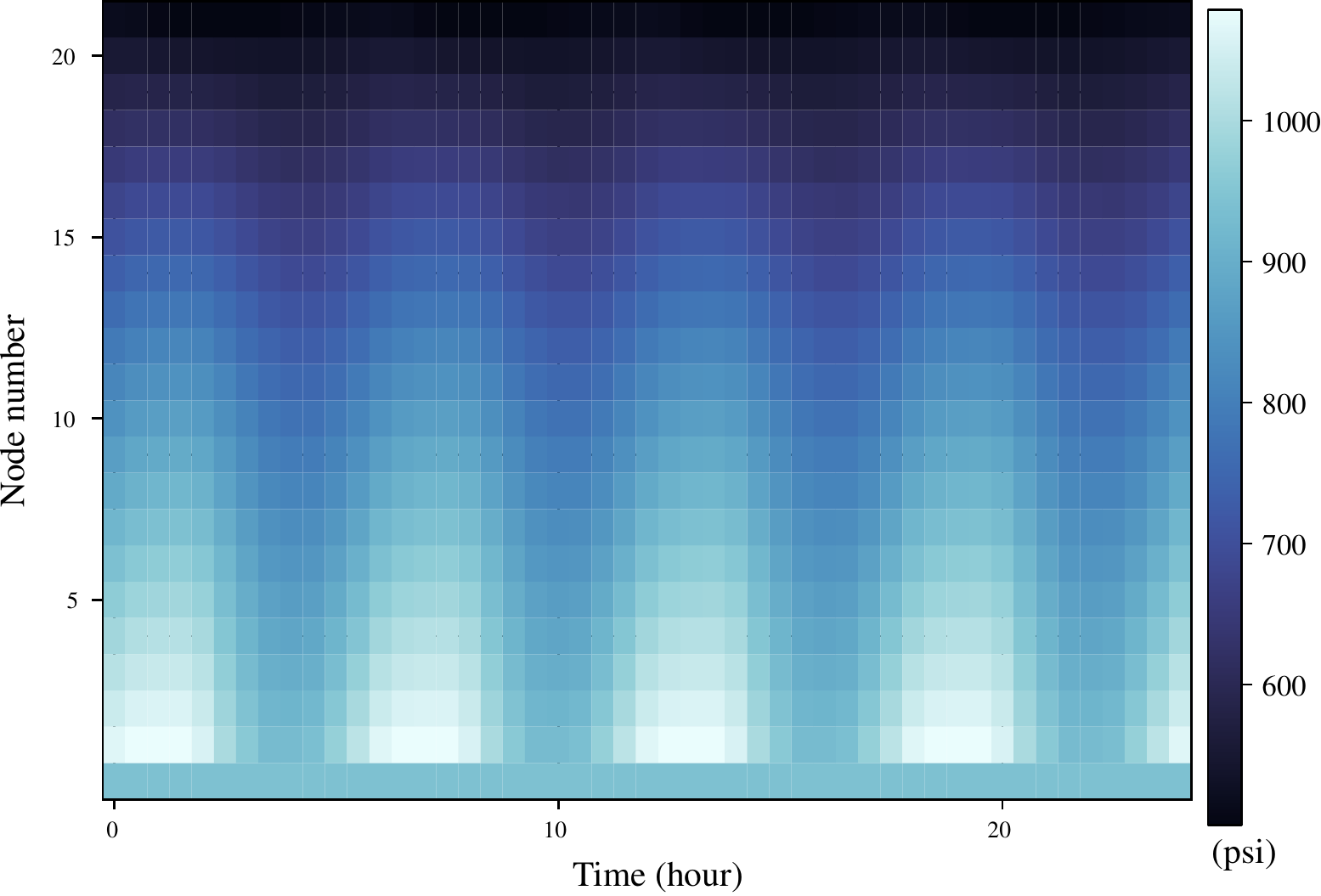}
    \caption{Simulated values for nodal pressures for the single-pipe case.  Each row in the image corresponds to the time-series of pressure values for one of 21 nodes in the refined graph with 5 km discretization of the 100 km pipe.}
    \label{fig:sim_pressure}
\end{figure}

\begin{figure}[ht!]
    \centering
    \includegraphics[scale=0.5]{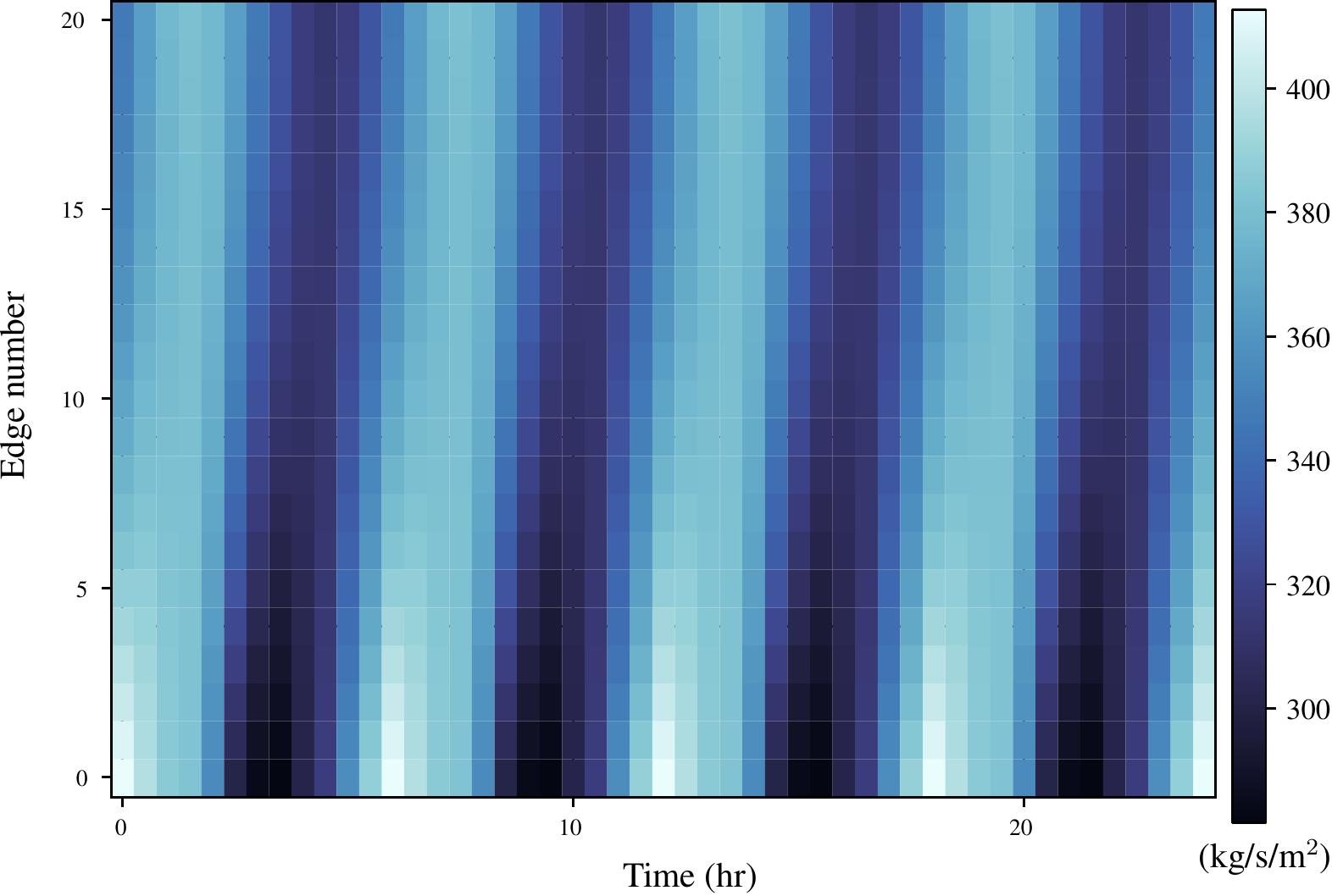}
    \caption{Simulated values for average edge mass flux for the single-pipe case.  Each row in the image corresponds to the time-series of mass flux values for one of the 20 edges in the refined graph with 5 km pipe segments.}
    \label{fig:sim_flux}
\end{figure}

The 4-node network contains a single loop. The schematic of the 4-node network is shown in the Fig. \ref{fig:4node}. The network contains two compressors and gas is being withdrawn at the nodes $2$, $3$ and $4$. The pressure at the slack node, $1$, is fixed to $500$ psi and gas is withdrawn at the withdrawal nodes in a time-periodic manner. Similar to the single-pipe case, the compression functions at the two compressors are time-periodically varied over the time horizon of $T = 24$ hours. 
\begin{figure}[ht!]
\centering
    \includegraphics[scale=0.75]{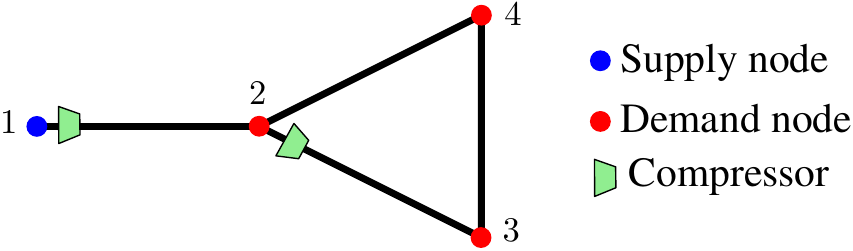}
    \caption{4-node network schematic.}
\label{fig:4node}
\end{figure}

The 25-node network is a tree with $5$ compressors and $24$ pipelines. The schematic of the network is shown in the Fig. \ref{fig:25node}. The pipelines in both the schematics in Fig. \ref{fig:4node} and \ref{fig:25node} are not to scale, they are presented to illustrate the topology of the network only. The slack pressure is fixed to $500$ psi. In the subsequent sections, we present computational results showing that the friction factor values can be estimated using noisy measurements. 
\begin{figure}[ht!]
\centering
    \includegraphics[scale=0.75]{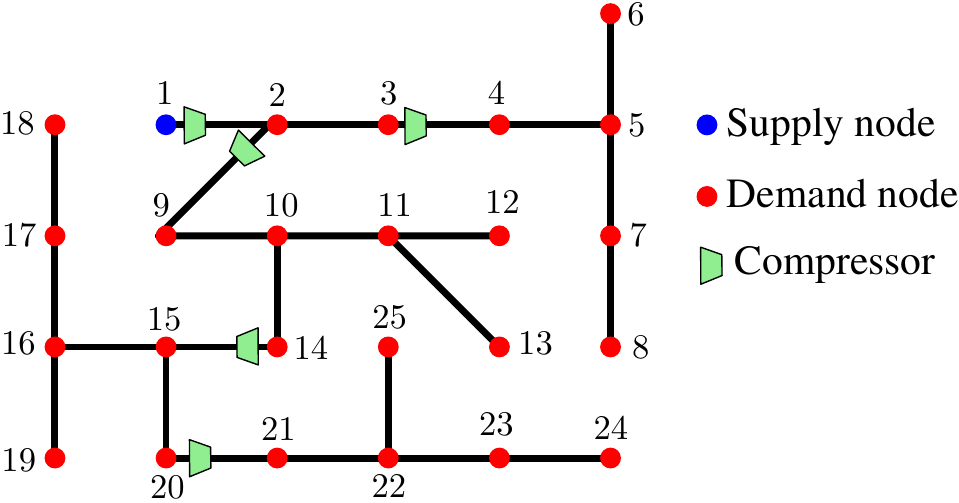}
    \caption{25-node network schematic.}
\label{fig:25node}
\end{figure}
Similar to the single pipe case, a simulation to obtain the nodal densities and average edge mass flux is performed with time-periodic withdrawal profiles and compression functions for each compressor. For networks with more complicated structure than a single pipe, synthesizing the withdrawal and compression functions (required as parameters) is not trivial. The form of these profiles is in principle not important, as long as they result in feasible pressures and flows throughout the estimation time horizon.  We obtain these parameter functions by solving an optimization problem subject to the same network flow model described above as constraints.  Specifically, they are obtained by solving optimal control problems for the example systems where the objective is to minimize compressor power, similar to formulations in previous studies \cite{Zlotnik2015,Mak2016}.  These time-series are also used as parameters in the ground truth simulation for the estimation case studies as well as synthetic measurement time-series after the addition of noise.  This guarantees that the data for the case studies are self-consistent and satisfy the stated inequality constraints on the state variables.  The pipe friction factor $\lambda$ is set to $0.011$ for the single pipe case and $0.01$ for the 4-node and 25-node networks. A spatial discretization of $5$ km is used for all the remaining computational experiments, and this is implemented by a spatial graph refinement where each edge $(i,j)$ is divided into the minimum number $n_{ij}$ of segments of equal length $\Delta_{ij}\leqslant L_{ij}/n_{ij}$.   

\subsection{Error performance of the state estimation problem} \label{subsec:error_results_state}
In this section, we present the results of the state estimation procedure for the problem defined in Eq. \eqref{eq:state_2} on the three test cases. To illustrate the computational effectiveness of the estimation problems, we use two relative error metrics: (i) average relative error in the state estimates and (ii) maximum relative error in the state estimates; the averages and maximum values are computed over nodal density estimates and edge flow-rates separately. For the single pipe case, we additionally present the the nodal density and edge mass flux profiles for various noise levels. For the 4-node and the 25-node networks we present the results of the error metrics for various noise levels.  

We choose three different measurement noise levels for the withdrawal and nodal pressure measurements i.e., we use an additive white Gaussian noise model with mean $0$ and a standard deviation of $0.5\%$, $1\%$, $1.5\%$, and $10\%$ of the true value obtained from the simulations. In the single pipe case, these noise levels correspond to a maximum error of $2.5$, $5$, $7.5$, and $50$ kg/s, respectively, in the withdrawal measurements; as for the nodal pressure measurements, the noise levels correspond to a maximum error $5$, $10$, $15$, and $100$ psi, respectively. From a practical stand-point, the uncertainty level of $10\%$ is quite high. However, simulations are performed for this noise level to evaluate whether the state and parameter estimation approach is feasible at such elevated noise levels, or whether the method breaks down.  We assume that the measurement noise is additive, and that we have withdrawal  and pressure measurements at every physical node in $\mathcal{V}_d$ where gas is being withdrawn.

The tables \ref{tab:1_1} -- \ref{tab:1_3} show the relative error metrics for at the chosen noise levels for the single pipe, 4-node, and 25-node networks, respectively.  Each value is obtained from a single instance of an estimation problem with synthetic measurement data. The column headings used in the tables are defined as:

%\begin{tabular}{p{.5cm}p{7cm}}
\noindent $\bm{nl}$: additive white Gaussian noise (as defined in the previous paragraph) added to the measurements, in percent.\\
\noindent $e^d_{\max}$: maximum relative error in all the estimates of the withdrawal values, in percent.\\
\noindent $e^p_{\max}$: maximum relative error in all the estimates of nodal pressures, in percent.\\
\noindent $e^{\varphi}_{\max}$: maximum relative error in all the estimates of the average mass flux at the edges, in percent. \\
\noindent $e^d_{\operatorname{avg}}$: average relative error in all the estimates of the withdrawal values, in percent.\\
\noindent $e^p_{\operatorname{avg}}$: average relative error in all the estimates of nodal pressures, in percent.\\
\noindent $e^{\varphi}_{\operatorname{avg}}$: average relative error in all the estimates of the average mass flux, in percent. 
%\end{tabular}

\begin{table}[htbp]
    \centering
    \caption{State estimation errors for the single pipe case.}
    \label{tab:1_1}
    \begin{tabular}{crrrrrr}
    \toprule
        $\bm{nl}$ & $e^d_{\max}$ & $e^p_{\max}$ & $e^{\varphi}_{\max}$ & $e^d_{\operatorname{avg}}$ & $e^p_{\operatorname{avg}}$ & $e^{\varphi}_{\operatorname{avg}}$\\ 
    \midrule
    $10$ & $17.50$ & $18.86$ & $16.03$ & $3.95$ & $0.83$ & $1.34$ \\
    $1.5$ & $3.58$ & $3.60$ & $3.33$ & $0.89$ & $0.25$ & $0.35$ \\
    $1.0$ & $2.74$ & $2.22$ & $2.52$ & $0.63$ & $0.22$ & $0.29$ \\
    $0.5$ & $1.65$ & $1.20$ & $1.51$ & $0.36$ & $0.13$ & $0.17$ \\
    \bottomrule
    \end{tabular} 
\end{table}

\begin{table}[htbp]
    \centering
    \caption{State estimation errors for the 4-node network.}
    \label{tab:1_2}
    \begin{tabular}{crrrrrr}
    \toprule
        $\bm{nl}$ & $e^d_{\max}$ & $e^p_{\max}$ & $e^{\varphi}_{\max}$ & $e^d_{\operatorname{avg}}$ & $e^p_{\operatorname{avg}}$ & $e^{\varphi}_{\operatorname{avg}}$\\ 
    \midrule
    $10$ & $17.44$ & $10.09$ & $53.04$ & $6.05$ & $1.61$ & $3.66$ \\
    $1.5$ & $4.94$ & $2.01$ & $5.17$ & $0.92$ & $0.31$ & $0.47$ \\
    $1.0$ & $2.23$ & $1.35$ & $3.86$ & $0.67$ & $0.24$ & $0.32$ \\
    $0.5$ & $1.46$ & $1.03$ & $2.81$ & $0.37$ & $0.24$ & $0.18$ \\
    \bottomrule
    \end{tabular}   
\end{table}

\begin{table}[htbp]
    \centering
    \caption{State estimation errors for the 25-node network.}
    \label{tab:1_3}
    \begin{tabular}{crrrrrr}
    \toprule
        $\bm{nl}$ & $e^d_{\max}$ & $e^p_{\max}$ & $e^{\varphi}_{\max}$ & $e^d_{\operatorname{avg}}$ & $e^p_{\operatorname{avg}}$ & $e^{\varphi}_{\operatorname{avg}}$\\ 
    \midrule
    $10$ & $86.42$ & $12.41$ & $329.92$ & $26.41$ & $1.73$ & $33.68$ \\
    $1.5$ & $16.84$ & $2.37$ & $262.98$ & $4.17$ & $0.26$ & $4.84$ \\
    $1.0$ & $11.51$ & $0.99$ & $153.11$ & $2.55$ & $0.18$ & $3.16$ \\
    $0.5$ & $5.17$ & $0.50$ & $90.27$ & $1.24$ & $0.09$ & $1.43$ \\
    \bottomrule
    \end{tabular}
\end{table}

The large relative errors in $\varphi$ estimates in the 25-node network case study occur because of flow reversals.  When the flow is near zero, relative errors are amplified significantly.  To remove the effect of flow reversals from computation of the maximum and average relative flow error, we include data for time points when the flow magnitude is above a threshold of 1 kg/s. As expected, the average relative errors in the state estimates and withdrawal estimates obtained by solving the state estimation problem in Eq. \eqref{eq:state_2} decrease with the decrease in noise variance. The average relative errors are fairly small for measurement uncertainty of 1.5\% or less, which indicates that the least squares approach is  effective in estimating all the states of the system. However, at an elevated uncertainty level of 10\%, the maximum and average relative errors in flow estimates are high for the 25-node case study, although the pressure estimates are quite reasonable, with less than 2\% average relative error.  Thus, the method in the test case estimates pressure more accurately than flow, and this trade-off could be calibrated by changing the weighting matrices $W_1$ and $W_2$ in the objective function \eqref{eq:state_2_obj}. The Fig. \ref{fig:s_pressure} and \ref{fig:s_flux} shows the absolute errors in the nodal density and average mass flux profiles at the edges for the single pipe case. We remark that for all the computational experiments in this paper, we have assumed that noisy nodal pressure and withdrawal measurements are available at all physical nodes where gas is being withdrawn. But in practice, this might not be the case. We relegate the development of state estimation for the flow of natural gas through a network using such sparse measurements to future work. In fact, observability of such general nonlinear system that have a network structure is itself an open problem in control theory. 

\begin{figure}[ht!]
    \centering
    \includegraphics[scale=0.5]{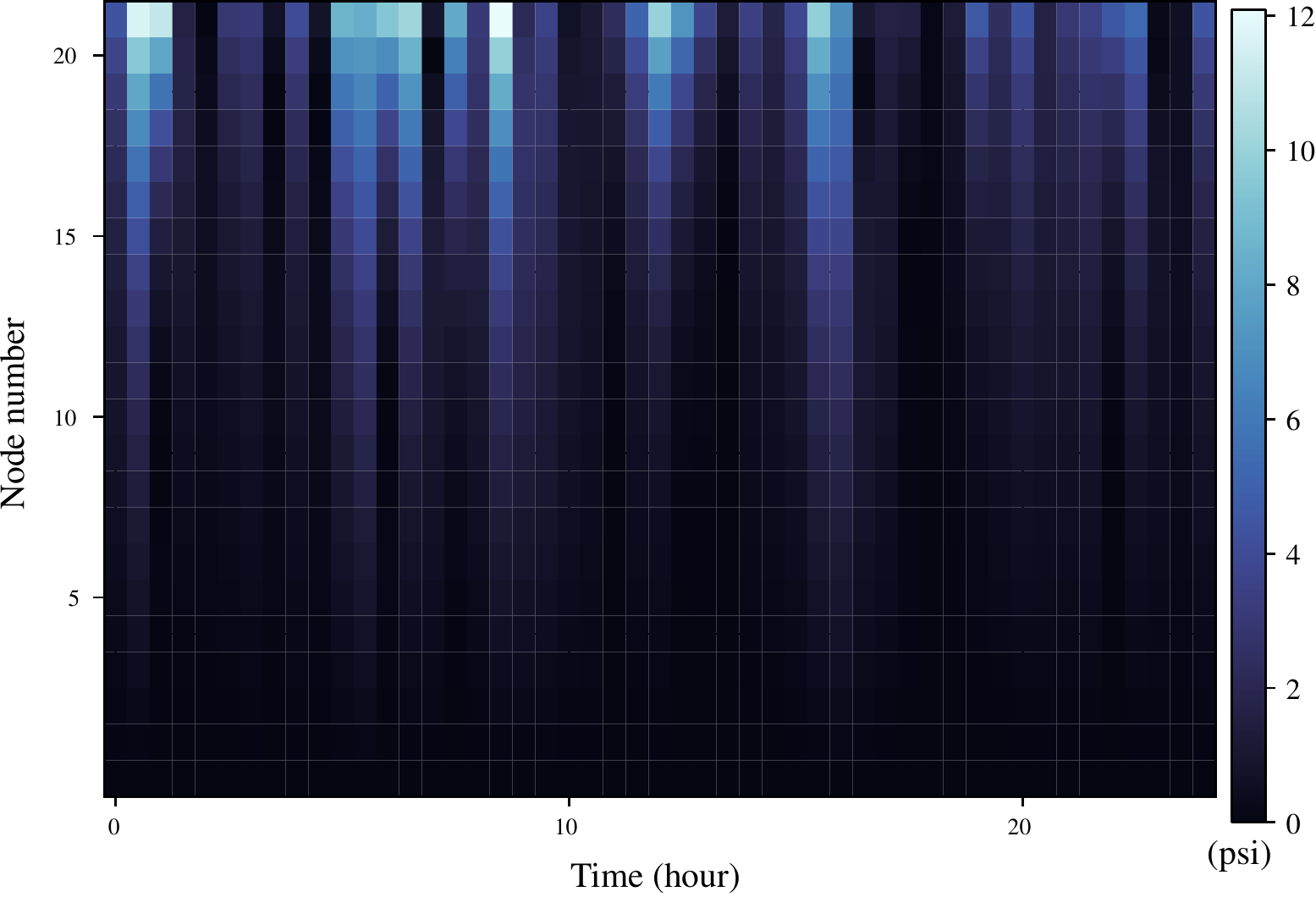}
    \caption{Absolute error in the nodal pressure estimates for the single pipe case for a noise level of $1.5\%$ in the measurements.  Each row in the image corresponds to the time-series of pressure estimates for one of 21 nodes in the refined graph with 5 km discretization of the 100 km pipe.}
    \label{fig:s_pressure}
\end{figure}

\begin{figure}[ht!]
    \centering
    \includegraphics[scale=0.5]{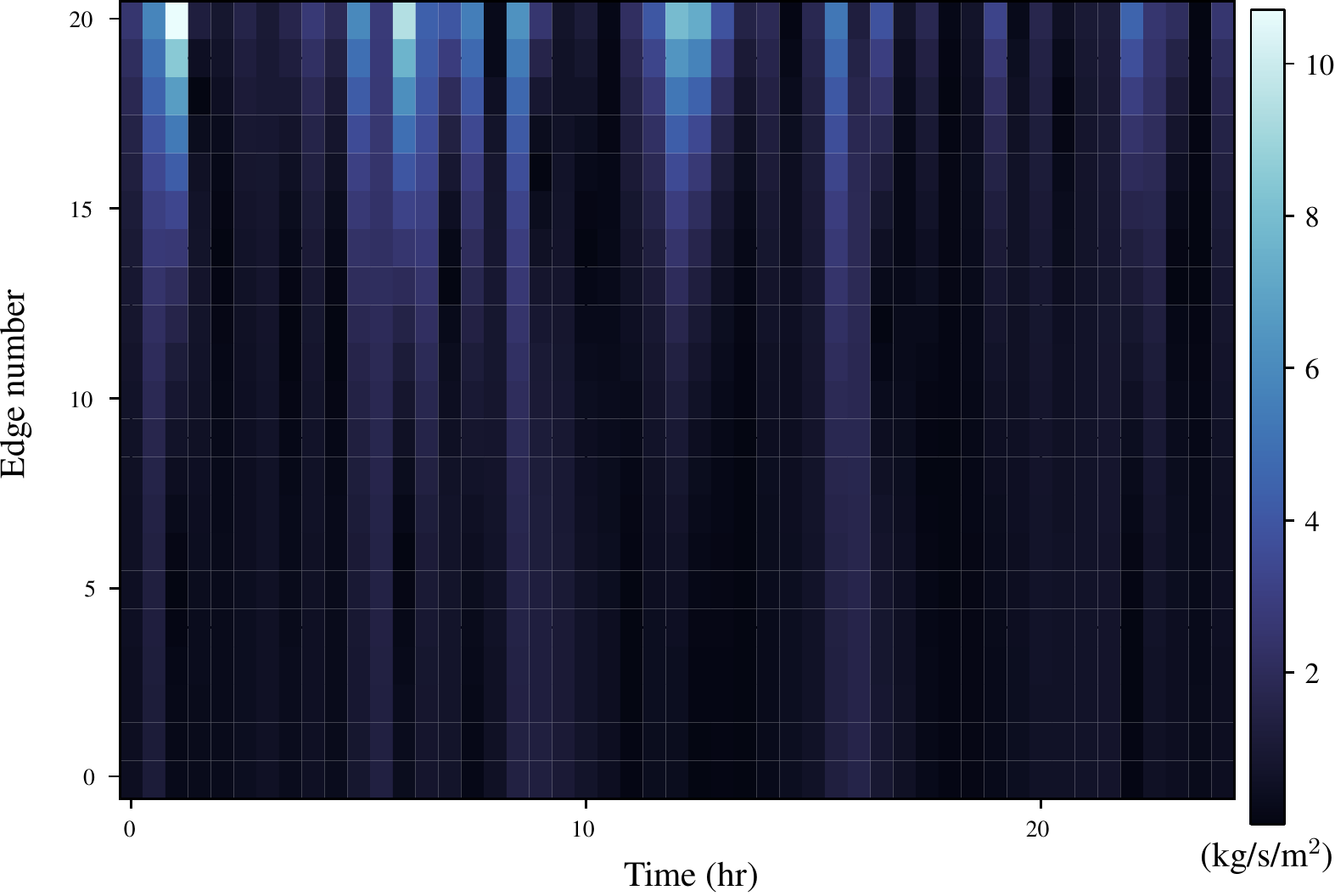}
    \caption{Absolute error in the average mass flux estimates at the edges for the single pipe case for a noise level of $1.5\%$ in the measurements. Each row in the image corresponds to the time-series of mass flux estimates for one of the 20 edges in the refined graph with 5 km pipe segments.}
    \label{fig:s_flux}
\end{figure}

\subsection{Error performance of the joint state  and parameter estimation problem} \label{subsec:error_results_state_param}

In this section, we present tables and plots similar to the tables \ref{tab:1_1} -- \ref{tab:1_3} and Fig. \ref{fig:s_pressure} and \ref{fig:s_flux}, respectively for the joint state  and parameter estimation problems. In addition, we also present the absolute error in friction factor estimates for each pipeline for each of the 3 test instances. The maximum and average relative errors in the state estimates and the withdrawal estimates is given by tables \ref{tab:2_1} -- \ref{tab:2_3} and the absolute error profiles for the single pipe case is shown in Fig. \ref{fig:sp_pressure} and \ref{fig:sp_flux}. The error trends in the state and withdrawal estimates are similar to the ones obtained by doing state estimation separately. Unlike the results for the state estimation, the relative error in the state estimates obtained via the formulation for joint state  and parameter estimation do not decrease at all times with decreasing noise levels; this is due to the fact that the error in the parameter estimates can in turn affect the error in the state estimates and these errors are related to each other in a nonlinear fashion (see Eq. \eqref{eq:phi_vec_f}).

As in the previous section, the average relative errors in the state estimates and withdrawal estimates obtained by solving the joing state and parameter estimation problem in Eq. \eqref{eq:state_2} decrease with the decrease in noise variance. The average relative errors are fairly small for measurement uncertainty of 1.5\% or less, which indicates that the least squares approach is effective in estimating all the states of the system. However, at an elevated uncertainty level of 10\%, the maximum and average relative errors in flow estimates are high for both the 4-node and 25-node case studies, although the relative error in pressure estimates is low, with less than 2\% average relative error.  Thus, the method in the test case estimates pressure more accurately than flow, and this trade-off could be calibrated by changing the weighting matrices $W_1$ and $W_2$ in the objective function \eqref{eq:state_2_obj}.

The parameter estimates for the 4-node and 24-node study are illustrated in Figures \ref{fig:p_4node} and \ref{fig:p_25node}, respectively.  We see that the error is quite high for even the low noise cases, but this error depends substantially on the pipe segment length, and we interpret the significance of this in the next section.  When performing the optimization of the joint state and parameter estimation problem, we place constraints on the parameter estimates at 50\% and 200\% of the true values (i.e., at 0.005 and 0.02), and we see that for the 10\% noise level these constraints are binding.  As a result, at this high uncertainty level the approach yields a mixed result; the pressure state estimation performs well, while the flow state estimation and the friction parameter estimation does not produce an acceptable outcome.  These results motivate additional work beyond this initial study to test the effect of uncertainty level and weighting of the objective function on state and parameter estimation accuracy.

\begin{table}[ht]
    \centering
        \caption{State estimation errors for the single pipe case.}
    \label{tab:2_1}
    \begin{tabular}{crrrrrr}
    \toprule
        $\bm{nl}$ & $e^d_{\max}$ & $e^p_{\max}$ & $e^{\varphi}_{\max}$ & $e^d_{\operatorname{avg}}$ & $e^p_{\operatorname{avg}}$ & $e^{\varphi}_{\operatorname{avg}}$\\ 
    \midrule
    $10$ & $18.19$ & $21.99$ & $16.85$ & $5.33$ & $1.31$ & $1.73$ \\
    $1.5$ & $2.73$ & $2.45$ & $2.47$ & $1.06$ & $0.19$ & $0.33$ \\
    $1.0$ & $2.42$ & $1.82$	& $2.24$ & $0.70$ & $0.14$ & $0.29$ \\
    $0.5$ & $1.29$ & $0.94$ & $1.16$ & $0.45$ & $0.09$ & $0.18$ \\
    \bottomrule
    \end{tabular}
\end{table}

\begin{table}[ht]
    \centering
        \caption{State estimation errors for the 4-node network.}
    \label{tab:2_2}
    \begin{tabular}{crrrrrr}
    \toprule
        $\bm{nl}$ & $e^d_{\max}$ & $e^p_{\max}$ & $e^{\varphi}_{\max}$ & $e^d_{\operatorname{avg}}$ & $e^p_{\operatorname{avg}}$ & $e^{\varphi}_{\operatorname{avg}}$\\ 
    \midrule
    $10$ & $25.88$ & $13.92$ & $88.88$ & $7.52$ & $1.69$ & $20.26$ \\
    $1.5$ & $3.14$ & $1.76$ & $37.14$ & $0.92$ & $0.38$ & $16.94$ \\
    $1.0$ & $1.67$ & $1.98$ & $7.80$ & $0.64$ & $0.33$ & $2.80$ \\
    $0.5$ & $1.51$ & $1.18$ & $3.04$ & $0.43$ & $0.21$ & $0.41$ \\
    \bottomrule
    \end{tabular}
\end{table}

\begin{table}[ht]
    \centering
        \caption{State estimation errors for the 25-node network.}
    \label{tab:2_3}
    \begin{tabular}{crrrrrr}
    \toprule
        $\bm{nl}$ & $e^d_{\max}$ & $e^p_{\max}$ & $e^{\varphi}_{\max}$ & $e^d_{\operatorname{avg}}$ & $e^p_{\operatorname{avg}}$ & $e^{\varphi}_{\operatorname{avg}}$\\ 
    \midrule
    $10$ & $79.08$ & $10.72$ & $229.73$ & $34.91$ & $2.00$ & $47.53$ \\
    $1.5$ & $23.53$ & $1.56$ & $189.33$ & $4.74$ & $0.69$ & $6.07$ \\
    $1.0$ & $11.10$ & $1.29$ & $124.44$ & $2.90$ & $0.22$ & $3.22$ \\
    $0.5$ & $6.64$ & $0.92$ & $82.85$ & $1.56$ & $0.12$ & $1.88$ \\
    \bottomrule
    \end{tabular}
\end{table}

\begin{figure}[ht!]
    \centering
    \includegraphics[scale=0.5]{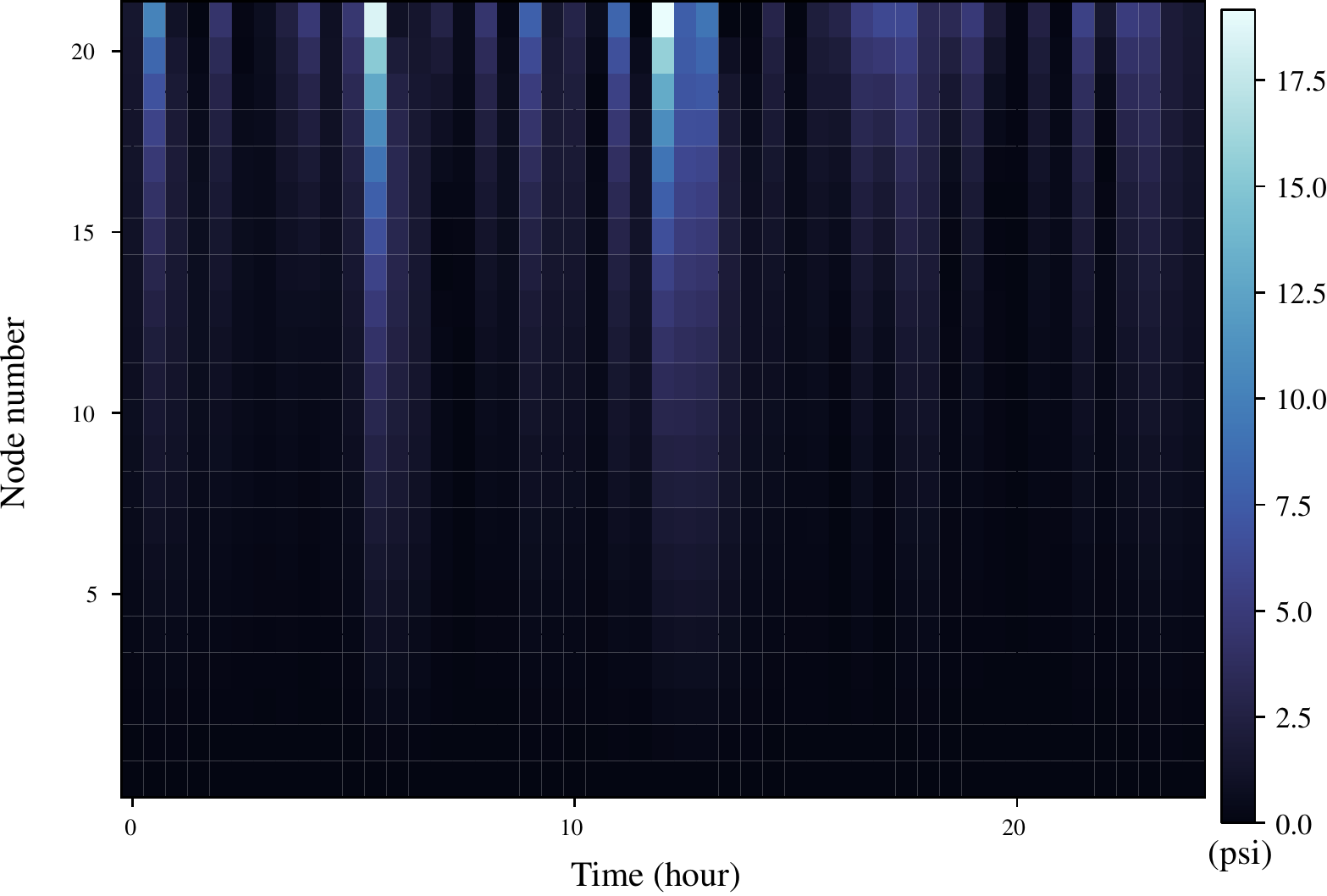}
    \caption{Absolute error in the nodal pressure estimates for the single pipe case for a noise level of $1.5\%$ in the measurements.  Each row in the image corresponds to the time-series of pressure estimates for one of 21 nodes in the refined graph with 5 km discretization of the 100 km pipe.}
    \label{fig:sp_pressure}
\end{figure}

\begin{figure}[ht!]
    \centering
    \includegraphics[scale=0.5]{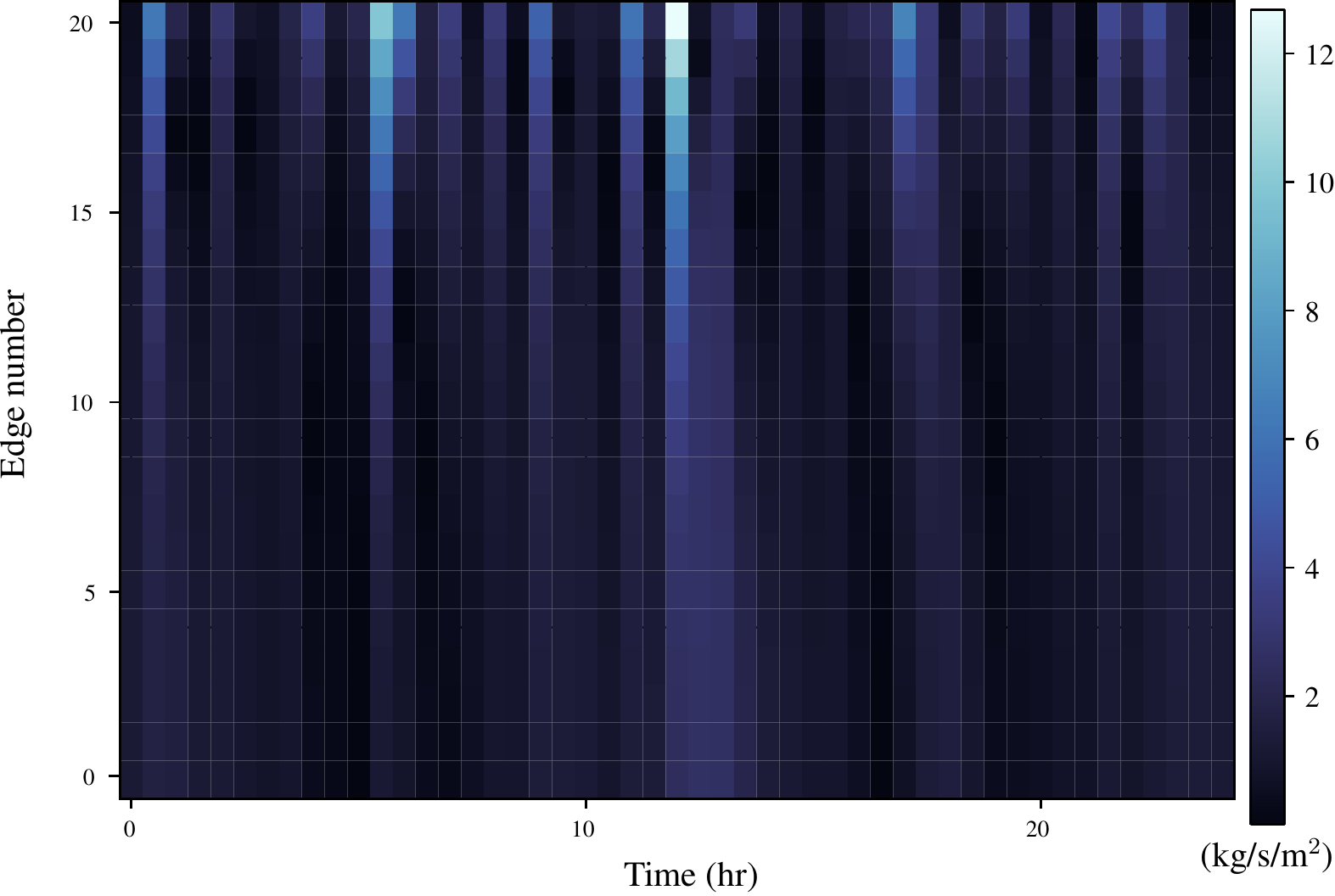}
    \caption{Absolute error in the average mass flux estimates at the edges for the single pipe case for a noise level of $1.5\%$ in the measurements. Each row in the image corresponds to the time-series of mass flux estimates for one of the 20 edges in the refined graph with 5 km pipe segments.}
    \label{fig:sp_flux}
\end{figure}

The table \ref{tab:3_1} shows the true and estimated value of the friction factor for the single pipe case. The Fig. \ref{fig:p_4node} and \ref{fig:p_25node} present a plot of the parameter estimates using a single run for the 4-node network and 25-node network, respectively. From the Fig. \ref{fig:p_4node} and \ref{fig:p_25node}, we observe that the parameter estimates are very sensitive the noise in the measurements which in turn affect the state estimates. 

\begin{table}
    \centering
    \caption{Pipeline friction factor estimates for the single pipe case}
    \label{tab:3_1}
    \begin{tabular}{crr}
    \toprule
        $\bm{nl}$ & $\lambda_{\operatorname{true}}$ & $\lambda_{\operatorname{estimated}}$ \\
    \midrule
    $10$ & $0.011$ & $0.0108$ \\
    $1.5$ & $0.011$ & $0.0110$ \\
    $1.0$ & $0.011$ & $0.0109$ \\
    $0.5$ & $0.011$ & $0.0110$ \\
    \bottomrule
    \end{tabular}
\end{table}

\begin{figure}
    \centering
    \includegraphics[scale=1]{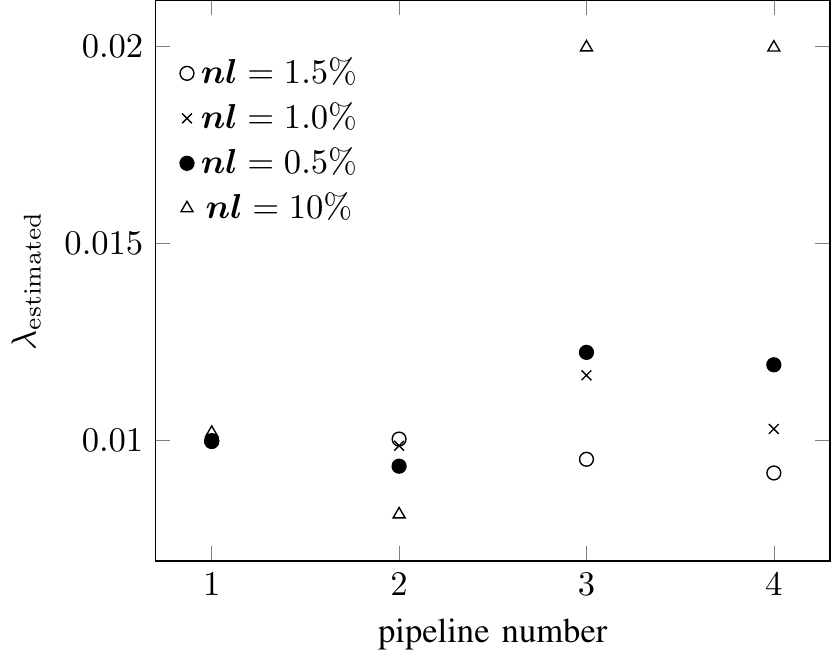}
%     \begin{tikzpicture}
%     \begin{axis}[xlabel={pipeline number},ylabel={$\lambda_{\operatorname{estimated}}$},
%     yticklabel style={
%         /pgf/number format/fixed,
%         /pgf/number format/precision=4
%     },
%     scaled y ticks=false,
%     xtick={0,1,...,4},
%     legend style={at={(0.02,0.75)},anchor=west,draw=none}
%     ]
%     \addlegendimage{only marks, mark=o}
%     \addlegendimage{only marks, mark=x}
%     \addlegendimage{only marks, mark=*}
%     \addlegendimage{only marks, mark=triangle}

%     % Graph column 1 versus column 0
%     \addplot[color=black, mark=o, draw=none] table[x index=0,y index=1,col sep=comma] {4node.dat};
%     \addlegendentry{$\bm{nl} = 1.5\%$}

%     % Graph column 2 versus column 0    
%     \addplot[color=black, mark=x, draw=none] table[x index=0,y index=2,col sep=comma] {4node.dat};
%     \addlegendentry{$\bm{nl} = 1.0\%$}
    
%     % Graph column 3 versus column 0
%     \addplot[color=black, mark=*, draw=none] table[x index=0,y index=3,col sep=comma] {4node.dat};
%     \addlegendentry{$\bm{nl} = 0.5\%$}
    
%     % Graph column 3 versus column 0
%     \addplot[color=black, mark=triangle, draw=none] table[x index=0,y index=4,col sep=comma] {4node.dat};
%     \addlegendentry{$\bm{nl} = 10\%$}

% \end{axis}
% \end{tikzpicture}
    \caption{Parameter estimates for the 4-node network. The $\lambda_{\operatorname{true}}$ for each pipeline is $0.01$}
    \label{fig:p_4node}
\end{figure}

\begin{figure}
    \centering
    \includegraphics[scale=1]{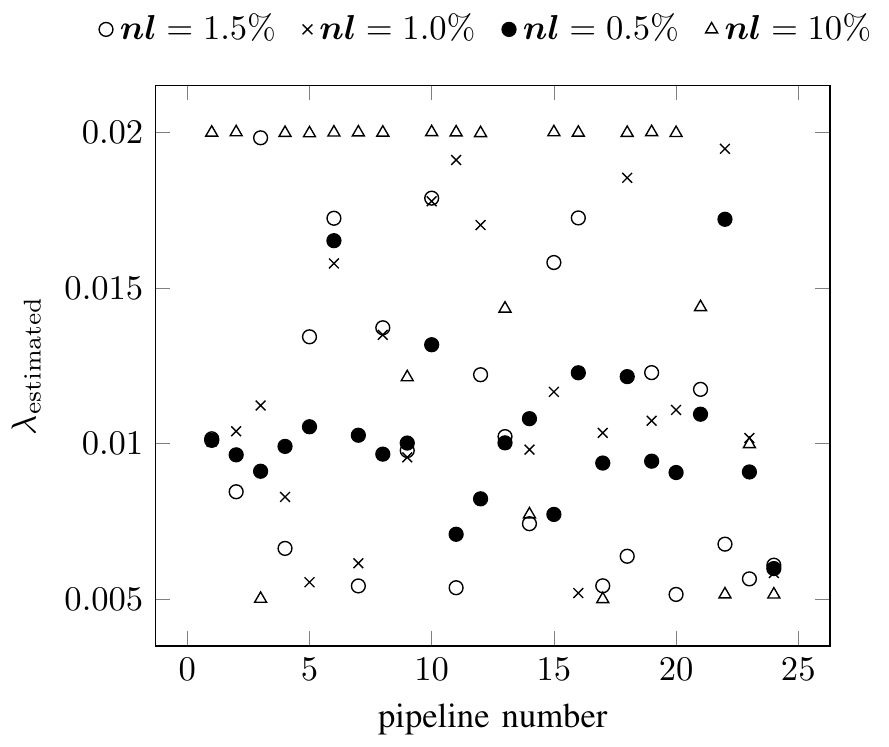}
%     \begin{tikzpicture}
%     \begin{axis}[xlabel={pipeline number},ylabel={$\lambda_{\operatorname{estimated}}$},
%     yticklabel style={
%         /pgf/number format/fixed,
%         /pgf/number format/precision=4
%     },
%     scaled y ticks=false,
%     %xtick={0,1,...,24},
%     legend style={at={(-0.1,1.1)},anchor=west,draw=none,legend columns=-1, /tikz/every even column/.append style={column sep=0.2cm}}
%     ]
%     \addlegendimage{only marks, mark=o}
%     \addlegendimage{only marks, mark=x}
%     \addlegendimage{only marks, mark=*}
%     \addlegendimage{only marks, mark=triangle}

%     % Graph column 1 versus column 0
%     \addplot[color=black, mark=o, draw=none] table[x index=0,y index=1,col sep=comma] {25node.dat};
%     \addlegendentry{$\bm{nl} = 1.5\%$}

%     % Graph column 2 versus column 0    
%     \addplot[color=black, mark=x, draw=none] table[x index=0,y index=2,col sep=comma] {25node.dat};
%     \addlegendentry{$\bm{nl} = 1.0\%$}
    
%     % Graph column 3 versus column 0
%     \addplot[color=black, mark=*, draw=none] table[x index=0,y index=3,col sep=comma] {25node.dat};
%     \addlegendentry{$\bm{nl} = 0.5\%$}
    
%     % Graph column 4 versus column 0
%     \addplot[color=black, mark=triangle, draw=none] table[x index=0,y index=4,col sep=comma] {25node.dat};
%     \addlegendentry{$\bm{nl} = 10\%$}

% \end{axis}
% \end{tikzpicture}
    \caption{Parameter estimates for the 25-node network computed for different uncertainty or noise levels. The $\lambda_{\operatorname{true}}$ for each pipe is $0.01$}
    \label{fig:p_25node}
\end{figure}

\subsection{Weighted relative bias in the parameter estimates} \label{subsec:bias}
We now present a measure, which we utilize to study the errors in the parameter estimates when performing multiple runs of the joint state and parameter estimation algorithm. For this study, we will present results only for the 25-node network. We perform $n=30$ runs of the algorithm with a fixed noise level ($1.5\%$) in all the measurements. For this noise level, let $\lambda_{\operatorname{estimated}}^{n,ij}$ denote the friction factor estimated for pipe $(i,j)$ and run $n$. Also, let $\lambda_{\operatorname{true}}^{ij}$ denote the true value of the friction factor corresponding to pipe $(i,j)$. Given these notations, we define the \emph{weighted relative bias} in the friction factor estimate for pipe $(i,j)$ as 
\begin{flalign}
e_{ij}^n = \frac{\left(\frac 1n \sum_{k=1}^n \lambda_{\operatorname{estimated}}^{k,ij} - \lambda_{\operatorname{true}}^{ij}\right)}{\lambda_{\operatorname{true}}^{ij}} \cdot \frac{L_{ij}}{L_{\max}} \label{eq:measure}
\end{flalign}
where, $L_{\max}$ is the length of the longest pipe in the network. We remark that $e_{ij}^n$ is weighted by the ratio of the length of the pipe $(i,j)$ to the maximum length of any pipe in the network. The rationale behind this weighting is that for shorter pipes, the relative error in friction factor estimates can potentially be much larger than those for the longer pipes. This behaviour is due to the assumption in Eq. \eqref{eq:lumping_condition} which states that the changes in density and mass flux values over small pipeline segments of length $L$ are fairly small. Hence, for any pipe $(i,j) \in \mathcal E$ with length $L_{ij} \approx L$, the effect of noise in \eqref{eq:dae1b} is predominant which in turn results in larger parameter estimate errors.  We conclude that accuracy in estimation of the friction factor of a very short pipe is less important than estimating that of a long pipe for modeling an entire pipeline.  The rationale behind the weighted accuracy metric in \eqref{eq:measure} is to weight the importance of accurate estimation for a given pipe with respect to its influence on accuracy of modeling overall pipeline system dynamics.  The Fig. \ref{fig:rb} shows the value of $e_{ij}^n$ plotted for each run of the algorithm and for each pipe $(i,j) \in \mathcal E$ in the 25-node network. 

\begin{figure}
 \centering
    \includegraphics[scale=0.55]{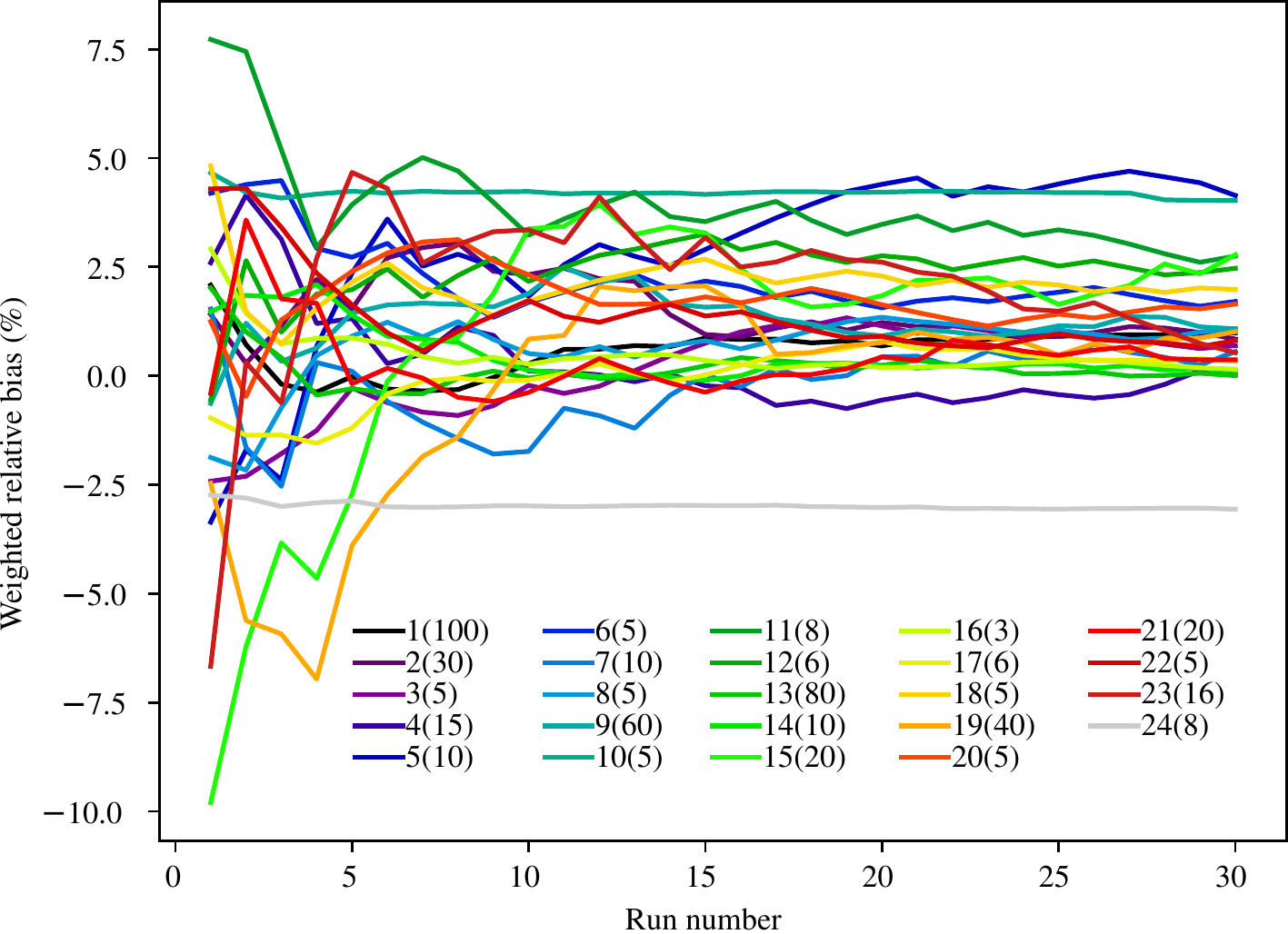}
    \caption{Weighted relative bias plot for every pipe $(i,j)\in \mathcal E$ in the 25-node network. The number in parentheses in the legend entries denote the length of the corresponding pipe in kilometers.}
    \label{fig:rb}
\end{figure}

\section{Conclusion} \label{sec:conclusion}
This manuscript develops formulations and algorithms for solving the state estimation and joint state and parameter estimation problems for the flow of natural gas through a large-scale network of pipelines with actuation by compressors. We have presented approaches to physical and engineering modeling, model reduction, control system modeling, and uncertainty modeling.  We have also derived a uniqueness result for the time-periodic boundary value problem for the discretized pipeline flow equations, which implies that pipeline state (pressure) measurements are not required for estimation when withdrawals are known exactly.  The proposed method is seen to be effective in computing state and parameter estimates in empirical studies on multiple test cases. In the computational case studies, the method consistently provides low error estimates of the pressures throughout the network under various conditions, although in the case of high uncertainty the performance of flow and friction parameter estimation degrades.

Future work would focus on testing these algorithms on actual time-series data obtained from real gas pipeline networks, effectively solving estimation problems using only sparse pressure measurements, and extending the algorithms robust to be robust to outliers in the measurements.  In addition, the modeling formalism developed here could be applied to develop leak detection techniques for sparsely instrumented systems.  Furthermore, the development of a Bayesian filtering approach for the DAE systems would be of general interest, and could be applied to reduced natural gas network dynamic models for use in pipeline system applications.

\section*{Acknowledgements}
This work was carried out as part of Project GECO for the Advanced Research Project Agency-Energy of the U.S. Department of Energy under Award No. DE-AR0000673. Work at Los Alamos National Laboratory was conducted under the auspices of the National Nuclear Security Administration of the U.S. Department of Energy under Contract No. DEAC52-06NA25396, and was supported by the Advanced Grid Modeling Research Program in the U.S. Department of Energy Office of Electricity. This work was also partially supported by the Center for Nonlinear Studies at Los Alamos National Laboratory.

\appendices
\section{Eigenvalue equation for Eq. \eqref{eq:pde_1}}\label{app:eig}
In this appendix, we present the eigenvalue equation for the system of partial differential equations in Eq. \eqref{eq:pde_1}. The eigenvalue equation is given by
\begin{flalign}
\begin{bmatrix}
\partial_x \rho \\ \partial_x \varphi 
\end{bmatrix} + 
\begin{bmatrix}
0 & 0 \\ 1 & 0
\end{bmatrix}
\begin{bmatrix}
\partial_t \rho \\ \partial_t \varphi 
\end{bmatrix} = 
\begin{bmatrix}
0 & \frac 1{a^2} \\ 1 & 0
\end{bmatrix}
\begin{bmatrix}
0 \\ -\frac{\lambda}{2D}\frac{\varphi |\varphi|}{\rho}
\end{bmatrix}.
\label{eq:eig}
\end{flalign}
The matrix $\left[ 0 ~0; 1~ 0 \right]$ in Eq. \eqref{eq:eig} has two repeated eigenvalues, both being $0$, indicating that original system of equations in Eq. \eqref{eq:pde_1} is parabolic in nature \cite{Renardy2006}. 

\section{Reduction of Eq. \eqref{eq:dae} to Eq. \eqref{eq:dae_final}} \label{app:dae}
In this appendix, we show that the Eqs. \eqref{eq:dae} can equivalently be represented by the DAE system in Eq. \eqref{eq:dae_final} using the graph-theoretic notation introduced in Sec. \ref{subsec:cs_model}. We remark that additional definitions and notations would be introduced in the derivation, as and when required. To that end, we first rewrite Eq. \eqref{eq:dae0b} in matrix form as follows:
\begin{flalign}
& \bm d = \bar{A}_d X \bar{\bm \varphi} + \ubar{A}_d X \ubar{\bm \varphi} & \label{eq:d_vec}  
\end{flalign}
where, $\bar{A}_d$ and $\ubar{A}_d$ are the positive and negative parts of the matrix $A_d$, respectively. We now define $\bm \Phi_{-} = \frac 12 (\bar{\bm \varphi} - \ubar{\bm \varphi})$. Using the definition of $\bm \Phi_{-}$, the Eq. \eqref{eq:d_vec} can be rewritten as 
\begin{flalign}
& \bm d = A_d X \bm \Phi + |A_d| X \bm \Phi_{-}. & \label{eq:d_veca}  
\end{flalign}
On the other hand, Eqs. \eqref{eq:dae0a}, \eqref{eq:dae0c}, and \eqref{eq:dae0e} together with the definition $\bm \Phi_{-}$ can be equivalently represented using the following matrix equation:
\begin{flalign}
& |B_s^\intercal| \dot{\bm s} + |B_d^\intercal| \dot{\bm \rho} = -4 \Lambda^{-1} \bm \Phi_{-}.  & \label{eq:rho_dynamics_vec}
\end{flalign}
Substituting Eq. \eqref{eq:d_veca} into Eq. \eqref{eq:rho_dynamics_vec} and eliminating $\bm \Phi_{-}$, we obtain
\begin{flalign*}
& |A_d| X \Lambda |B_d^\intercal|\dot{\bm \rho} = 4(A_d X \bm \Phi - \bm d) - |A_d| X \Lambda |B_s^\intercal| \dot{\bm s}  &
\end{flalign*}
which is the same as Eq. \eqref{eq:dae1a}. Then, Eq. \eqref{eq:dae0d} can be rewritten as 
\begin{flalign}
& \ubar{\rho}_{ij}^2 - \bar{\rho}_{ij}^2 = -\frac{\lambda \ell_0 L}{4D_{ij}} \Phi_{ij}|\Phi_{ij}|, \, \forall (i,j) \in \hat{\mathcal E}. & \label{eq:phi_vec}
\end{flalign}
Using Eq. \eqref{eq:dae0c} and the definitions $B$, $\Lambda$, and $K$ the above equation can be equivalently written in matrix form as
\begin{flalign*}
&  \Lambda K \bm \Phi \odot \bm \Phi =  -B^\intercal \bm \rho^N \odot |B^\intercal| \bm \rho^N, &
\end{flalign*}
completing the derivation. 

\bibliographystyle{plain}
\bibliography{spe.bib}

\end{document}